\documentclass[lettersize,journal]{IEEEtran}

\usepackage{amsmath,amsfonts,amssymb}
\usepackage{amsthm}
\usepackage{algorithm}
\usepackage{algorithmic}
\usepackage{array}
\usepackage{booktabs}
\usepackage{cite}
\usepackage{graphicx}
\usepackage{subcaption}
\usepackage{textcomp}
\usepackage{url}
\usepackage{xcolor}
\usepackage{soul}
\usepackage{hyperref}
\usepackage{cleveref}
\usepackage{balance}

\theoremstyle{plain}

\newtheorem{lemma}{Lemma}

\theoremstyle{definition}
\newtheorem{definition}{Definition}
\newtheorem{assumption}{Assumption}
\theoremstyle{remark}
\newtheorem{remark}{Remark}

\newcommand{\E}{\mathbb{E}}
\newcommand{\Prob}{\mathbb{P}}
\newcommand{\pos}[1]{\left[#1\right]^+}

\newcommand{\safeincludegraphics}[2][]{%
  \IfFileExists{#2}{\includegraphics[#1]{#2}}{%
    \fbox{\parbox[c][3.0cm][c]{0.90\linewidth}{\centering
      \textcolor{red}{\textbf{[FIGURE TODO:} Missing figure: \texttt{\detokenize{#2}}\textbf{]}}}}%
  }%
}

\crefname{figure}{Fig.}{Figs.}
\Crefname{figure}{Fig.}{Figs.}

\crefname{table}{Table}{Tables}
\Crefname{table}{Table}{Tables}

\crefname{section}{Section}{Sections}
\Crefname{section}{Section}{Sections}

\crefname{subsection}{Section}{Sections}
\Crefname{subsection}{Section}{Sections}

\crefname{subsubsection}{Section}{Sections}
\Crefname{subsubsection}{Section}{Sections}

\crefname{appendix}{Appendix}{Appendices}
\Crefname{appendix}{Appendix}{Appendices}

\crefname{equation}{equation}{equations}
\Crefname{equation}{Equation}{Equations}

\crefformat{equation}{(#2#1#3)}
\Crefformat{equation}{Equation~(#2#1#3)}

\crefrangeformat{equation}
  {(#3#1#4)--(#5#2#6)}
\Crefrangeformat{equation}
  {Equations~(#3#1#4)--(#5#2#6)}

\crefmultiformat{equation}
  {(#2#1#3)}
  { and~(#2#1#3)}
  {, (#2#1#3)}
  {, and~(#2#1#3)}

\Crefmultiformat{equation}
  {Equations~(#2#1#3)}
  { and~(#2#1#3)}
  {, (#2#1#3)}
  {, and~(#2#1#3)}

\crefrangemultiformat{equation}
  {(#3#1#4)--(#5#2#6)}
  { and~(#3#1#4)--(#5#2#6)}
  {, (#3#1#4)--(#5#2#6)}
  {, and~(#3#1#4)--(#5#2#6)}

\Crefrangemultiformat{equation}
  {Equations~(#3#1#4)--(#5#2#6)}
  { and~(#3#1#4)--(#5#2#6)}
  {, (#3#1#4)--(#5#2#6)}
  {, and~(#3#1#4)--(#5#2#6)}

\crefname{theorem}{Theorem}{Theorems}
\Crefname{theorem}{Theorem}{Theorems}

\crefname{lemma}{Lemma}{Lemmas}
\Crefname{lemma}{Lemma}{Lemmas}

\crefname{algorithm}{Algorithm}{Algorithms}
\Crefname{algorithm}{Algorithm}{Algorithms}

\crefname{proposition}{Proposition}{Propositions}
\crefname{corollary}{Corollary}{Corollaries}
\crefname{definition}{Definition}{Definitions}
\crefname{assumption}{Assumption}{Assumptions}
\crefname{remark}{Remark}{Remarks}

\begin{document}

\title{A Smallest-Need-First Job Scheduling Framework with Adaptive Optimization of Idle Node Counts for Energy-Efficient HPC Systems}

\author{Reza Pulungan, Raka Satya Prasasta, Santana Yuda Pradata, Mursalim, Hiroyuki Takizawa, \\ Muhammad Alfian Amrizal%
\thanks{R. Pulungan, S. Y. Pradata, Mursalim, and M. A. Amrizal are with the Department of Computer Science and Electronics, Universitas Gadjah Mada, Yogyakarta, Indonesia.}%
\thanks{R. S. Prasasta is with the Faculty of Industrial Technology, Universitas Ahmad Dahlan, Yogyakarta, Indonesia.}%
\thanks{H. Takizawa is with the Cyberscience Center, Tohoku University, Miyagi, Japan.}%
\thanks{This work has been submitted to the IEEE Transactions on Parallel and Distributed Systems for possible publication.}
}

\markboth{}%
{Pulungan \MakeLowercase{\textit{et al.}}: SNF-ICON}

\maketitle

\begin{abstract}
    Power-state management in high-performance computing (HPC) clusters must reduce idle energy without excessive wake-up delays for rigid parallel jobs. This paper presents SNF-ICON, an event-driven controller combining smallest-need-first (SNF) gang scheduling, predictive wake timing, and adaptive warm-spare control. At each scheduler invocation, recent interarrival and completed-service samples are screened for sufficiency, exponential-like variability, low lag-one autocorrelation, and acceptable Kolmogorov–Smirnov distance. Rejected or data-sparse windows use SNF+IPM (Intelligent Power Manager), whereas accepted windows activate release prediction and an exponential next-event model. Warm-spare optimization is applied only when queue, event, and arrival-recency conditions permit, balancing estimated waiting and non-compute energy over a timeout-capped horizon. We evaluate four DAS2 trace segments and a generated Markovian workload on AOBA-derived 64-node models, plus SDSC Blue on an AOBA-derived 1152-node model. SNF-ICON is compared with SNF+IPM and First Come First Served (FCFS) + backfilling with IPM. It reduces average waiting time relative to the FCFS-based baseline in all six cases and remains close to at least one heuristic energy baseline in five. The generated workload spends substantial time in ICON mode, whereas DAS2 workloads operate mainly in fallback. Furthermore, cross-platform results show strong dependence on node-transition and power models. Thus, no single policy or parameter set works best in every case.
\end{abstract}


\section{Introduction}
\label{sec:introduction}

\IEEEPARstart{H}{igh}-performance computing (HPC) systems use large amounts of energy, and a single installation may require 20~MW or more. At the German HPC centers surveyed by Suarez et al. \cite{estela2025energyaware}, electricity and cooling made up between 12\% and 50\% of the systems' total cost of ownership. One source of avoidable energy use is idle compute capacity: nodes may stay powered on for long periods without running jobs. Energy-aware HPC schedulers therefore often use dynamic power management, such as lowering operating frequencies or switching off idle nodes~\cite{kocot2023energy}.

Production systems have also shown that idle-node power management can save substantial energy. At the J{\"u}lich Supercomputing Centre, the Slurm workload manager~\cite{yoo2003slurm} was configured to power down idle nodes. After this change, the annual electricity use of the JUSUF system fell from about 660~MWh in 2022 to 495~MWh in 2023, a reduction of 25\%~\cite{estela2025energyaware}. This result shows the practical value of controlling node power states. It also raises an important question: how many nodes can be switched off without making jobs wait too long?

Switching off too many nodes can increase job waiting time because a sleeping node must wake before it can run a job. Keeping too many nodes active reduces this delay but wastes energy while they are idle. This trade-off is especially important for rigid parallel, or gang-scheduled, jobs. A job that requests $r$ nodes can start only when all $r$ nodes are available at the same time. Job order and node power states should therefore be managed together~\cite{obrien2017, maiterth2018, gandhi2013}.

Existing methods include idle-timeout rules \cite{hu2009power, kammeyer2025slurm}, reinforcement-learning controllers \cite{liu2017hierarchical, farahnakian2014energy}, and scheduler--power-manager heuristics \cite{dupont2020energy, castan2024pareto}. Learning-based controllers can learn complex switch-on and switch-off actions, but they have several practical difficulties. Their reward must balance waiting time and energy, training needs a large amount of data and interaction, and the learned policy may work poorly when the workload or data quality changes~\cite{khasyah2022, budiarjo2025, sutton1998}. PSAS+IPM~\cite{prasasta2026} avoids offline training by using clear scheduling and power-management rules, but it still uses First Come First Served (FCFS) with Extensible Argonne Scheduling sYstem (EASY) backfilling~\cite{utilpred,lifka1995}. EASY backfilling is an aggressive backfilling policy in which later jobs may fill gaps in the schedule, provided that they do not delay the reserved start time of the head-of-queue job.

The scheduling method matters because power management cannot remove delays caused by job order. In a large multiserver-job model, Hong and Wang \cite{hong2024journal,hong2024} showed that FCFS mean waiting time grows worse than a lower bound that applies to all policies, while Smallest-Need-First (SNF) reaches that lower bound in growth rate. Their theorem does not directly cover EASY backfilling, finite non-preemptive HPC traces, or node wake-up times. Even so, it gives a clear reason to compare an SNF-based scheduler with an FCFS-based scheduler. In this paper, a job's need is its requested node count, and jobs are not preempted.

Changing the queue order does not solve the power-state problem by itself. A method such as PSAS+IPM can plan node transitions for queued jobs, but it uses requested runtimes, which may differ greatly from actual runtimes \cite{utilpred, tsafrir2007}. Planning for queued jobs also does not decide how much capacity should remain ready for future jobs. With perfect future information, a controller could wake exactly the required nodes so that they become ready when the next job arrives. In practice, future arrival times and node requests are unknown. The controller must therefore choose how many warm spare nodes to keep: too few cause wake-up delay, while too many waste idle energy.

This paper proposes SNF-ICON, an event-driven method that combines SNF scheduling, predicted wake times for queued jobs, and adaptive control of warm spare nodes for future jobs. Its main decision is the number of spare nodes to keep ready. The one-step calculation assumes exponential times between arrivals and completions, but real workloads do not follow this model at all times. We therefore check recent interarrival and completed-service data before using the calculation. When the recent data pass the check, the ICON controller predicts job releases and selects a warm-spare target. When the data fail the check, or there are too few samples, the system uses the SNF+IPM fallback. This limits model-based planning to periods when the recent workload fits the model closely enough. Furthermore:
\begin{itemize}
    \item We compare SNF with FCFS+backfilling-based power-management methods for rigid parallel jobs, motivated by a theoretical result on mean waiting time for multiserver jobs.
    \item We develop a one-step method for choosing how many warm spare nodes to keep ready.
    \item We define a fallback rule that uses SNF+IPM when recent data fail the workload check or contain too few samples. ICON-specific release prediction and spare-node selection are used only when the recent data pass the check.
    \item We predict job runtimes using an adaptive log ratio of actual to requested runtime. For running jobs, we estimate the remaining time with a conditional log-normal model. We combine these predictions with EMA estimates of arrival rate, service rate, and requested-node counts.
    \item We evaluate the method against FCFS/B+IPM and SNF+IPM, including analyses of the workload check, fallback, arrival-recency gate, platform, workload, parameters, and large-job waiting time.
\end{itemize}

The rest of the paper reviews related power-management methods, explains the choice of SNF and warm spare nodes, defines the system model and SNF-ICON method, and presents the implementation and evaluation.

\section{Background and Motivation}
\label{sec:background}

\subsection{Power management above the batch scheduler}

Switching off an idle node can save energy, but switching it off and waking it later also takes time and energy. A common rule is an idle timeout: a node stays active for a fixed time after becoming idle and is switched off if no job uses it. Benoit et al. \cite{benoit2017shutdown} studied shutdown policies that consider transition time, transition energy, power limits, and renewable energy. They used measurements from the Taurus cluster to model node transitions. These rules are easy to understand and deploy, but they react only after a node becomes idle. They do not decide how many nodes should stay ready for an unknown future job.

Reinforcement-learning (RL) power managers try to learn this decision from past workloads. Khasyah et al. \cite{khasyah2022} used an advantage actor--critic agent to choose switch-on and switch-off actions, and Budiarjo et al. \cite{budiarjo2025} used curriculum learning to improve training. These methods can learn complex policies, but they are difficult to deploy for the goals studied here. Waiting time and energy must be combined in the reward, training requires substantial data and time, and a policy trained on one system may not work well when the workload changes, events are sparse, user runtime estimates are inaccurate, or the job mix changes. Both methods also keep FCFS with backfilling as the scheduler, so they change the power manager without reconsidering job priority.

PSAS+IPM is a rule-based alternative that does not require training \cite{prasasta2026}. The power-state-aware scheduler (PSAS) includes node power states when dispatching and planning jobs, while the intelligent power manager (IPM) controls wake-up and shutdown actions. It can prepare nodes for queued jobs, but it predicts node release times from requested runtimes. Requested runtimes are useful, but they often differ from actual runtimes, so past execution data may improve the estimates \cite{utilpred, tsafrir2007}. PSAS+IPM also does not choose a warm-spare target for future jobs. Like the learning-based methods above, its original scheduler uses FCFS with EASY backfilling.

\subsection{Why reconsider FCFS+backfilling?}

FCFS with backfilling is widely used because it mostly keeps jobs in arrival order while allowing a later job to use idle resources when doing so does not delay a protected earlier job \cite{lifka1995, utilpred, tsafrir2007}. However, it does not minimize average waiting time for rigid parallel jobs. A rigid job is a multiserver job because it needs several nodes at the same time and cannot start with only part of its requested capacity \cite{feitelson1992, lublin2003}.

Hong and Wang \cite{hong2024} studied a multiserver-job system with $n_{\mathrm{HW}}$ servers under polynomial scaling. In their setting, the maximum server need is $\ell_{\max}=n_{\mathrm{HW}}^{\gamma_{\mathrm{HW}}}$, the spare capacity is $\delta_{\mathrm{HW}}=n_{\mathrm{HW}}^{\alpha_{\mathrm{HW}}}$, and the total arrival rate is $\Theta(n_{\mathrm{HW}})$, where $0\le\gamma_{\mathrm{HW}}<\alpha_{\mathrm{HW}}<(1+\gamma_{\mathrm{HW}})/2$. Under their heavy-traffic and maximum-need assumptions, and their commonness assumption for the SNF upper bound, the waiting-time growth rates are
\begin{align}
\E[T_{\mathrm{wait}}^{\mathrm{FCFS}}]&=\Theta\!\left(n_{\mathrm{HW}}^{\gamma_{\mathrm{HW}}-\alpha_{\mathrm{HW}}}\right),
\nonumber \\
\E[T_{\mathrm{wait}}^{\pi}]&=\Omega\!\left(n_{\mathrm{HW}}^{-\alpha_{\mathrm{HW}}}\right)
\quad \text{for every policy }\pi,
\label{eq:waiting-lower-bound}
\end{align}
and
\begin{equation}
\E[T_{\mathrm{wait}}^{\mathrm{SNF}}]=\Theta\!\left(n_{\mathrm{HW}}^{-\alpha_{\mathrm{HW}}}\right).
\label{eq:snf-achieves-bound}
\end{equation}
These results show that SNF reaches the lower-bound growth rate, while FCFS has a worse growth rate when $\gamma_{\mathrm{HW}}>0$, meaning that the largest job size grows with the system.

Hong and Wang's SNF model allows preemption and assumes that all servers are always available. Their theorem does not include node power states, EASY backfilling, finite traces, non-preemptive jobs, or node transition delays. We therefore use the theorem only as motivation for testing SNF. We do not claim that SNF is optimal for the system in this paper. The comparison with FCFS/B is based on experiments.

\subsection{Why future jobs require warm spare nodes?}

SNF decides which queued job is considered first, but it does not decide when sleeping nodes should wake or how many idle nodes should stay ready for future jobs. Queued jobs provide requested node counts and requested runtimes, so the scheduler can plan node releases and wake-ups for them. This plan is not exact because requested and actual runtimes may differ, but it is based on jobs that are already known.

The harder problem begins after all queued jobs have been handled. With perfect knowledge of the next job's arrival time and node request, a controller could wake exactly the required number of sleeping nodes so that they become ready when the job arrives. Waking them too late increases job waiting time, while waking them too early wastes energy. Because the next job is unknown, the controller instead keeps a chosen number of \emph{warm spare nodes}. These are active-idle nodes that can run a new job without a wake-up delay.

The main question is how many warm spare nodes to keep at each decision time. The proposed method estimates whether the next event will be a job arrival or a job completion, then balances wake-up-related waiting time against non-compute energy. Because this calculation uses a Markovian approximation, the controller first checks whether recent workload data fit that approximation closely enough. It prepares warm spare nodes when the data pass the check and uses SNF+IPM otherwise.

\section{System Model}
\label{sec:system}

\subsection{Cluster, jobs, and node states}

\begin{definition}[Node state partition]
\label{def:state-partition}
At time $t$, the node set of a cluster with $N\in\mathbb{N}_{>0}$ nodes is partitioned into
\begin{equation}
\mathcal{C}(t),\;\mathcal{I}(t),\;\mathcal{U}(t),\;\mathcal{D}(t),\;\mathcal{S}(t),
\end{equation}
representing computing, active-idle, switching-on, switching-off, and sleeping nodes, respectively. The five sets are pairwise disjoint, and their union is the complete node set.
\end{definition}

Each node $i$ is characterized by nonnegative state-dependent power values $P_i^{\mathrm{idle}}$, $P_i^{\mathrm{sleep}}$, $P_i^{\uparrow}$, and $P_i^{\downarrow}$, and by nonnegative expected transition times $L_i^{\uparrow}$ and $L_i^{\downarrow}$. The superscripts $\uparrow$ and $\downarrow$ indicate the switching-on and switching-off states, respectively.

\begin{definition}[Rigid job]
\label{def:rigid-job}
A job $j$ is described by an identifier $z_j$, a requested node count $r_j\in\{1,\ldots,N\}$, and a requested runtime $\tau_j>0$. It may start only when all $r_j$ nodes are available simultaneously \cite{feitelson1992,lublin2003}. The waiting queue at time $t$ is denoted by $\mathcal{Q}(t)$.
\end{definition}

\subsection{Smallest-need-first scheduling}

\begin{definition}[Smallest-need-first priority]
\label{def:snf-priority}
For waiting jobs $j$ and $k$, write $j\prec_{\mathrm{SNF}}k$ when
\begin{equation}
(r_j,\tau_j,z_j)<_{\mathrm{lex}}(r_k,\tau_k,z_k).
\label{eq:snf-key}
\end{equation}
Jobs are ordered first by requested node count. Requested runtime breaks ties between jobs of the same size, and the identifier breaks any remaining ties. In this study, job identifiers are assigned in ascending order based on their arrival.
\end{definition}

At each scheduling call, the queue is sorted using~\cref{eq:snf-key} and checked from the first job. A job starts immediately when at least $r_j$ suitable idle nodes are available. The check can stop at the first job that does not fit because every later job requests at least as many nodes and therefore cannot fit during the same call.

\section{The Markovianity-Gated SNF-ICON Heuristic}
\label{sec:method}


The proposed method must determine the number of warm spare nodes to prepare for future jobs. The one-step spare objective is tractable when the next arrival and completion are each Markovian. Real-life scenario workloads are not generally Markovian over an entire trace, but shorter recent windows may be sufficiently close to a Markovian system for a time-local decision. We therefore use a recent-window screen to choose between two operating modes at every scheduler invocation. Accepted windows, those deemed sufficiently similar to Markovian systems, use the proposed SNF-ICON spare calculation. Meanwhile, rejected or data-sparse windows use the SNF+IPM path, which combines SNF scheduling and queued-job wake planning with IPM but omits spare nodes planning.

Throughout this section, $t$ denotes the current scheduler time, $N$ is the total number of cluster nodes, $|\cdot|$ denotes the set cardinality, $\mathbf{1}\{\cdot\}$ is an indicator function, and $\pos{y}=\max\{y,0\}$. A hat denotes an estimated quantity, while a superscript ``$+$'' denotes the value after incorporating the newest observation. Unless otherwise stated, all time quantities use the simulator's time unit, which is in seconds.

\subsection{Markovianity of the observed workload}
\label{subsec:markov-screen}

The spare node objective metric requires a model that takes into account a race between the next job arrival and the next job completion. Consequently, the screen examines the two time series that determine those clocks: interarrival time intervals that describe the arrival process, and completed-job execution times that describe the service process. Meanwhile, requested-node counts are modeled separately through an empirical probability mass function, and they need not themselves be exponentially distributed. Screening both temporal series reduces the risk of applying a Markovian next-event model when either exhibits strong non-Markovian behavior.

Interarrival intervals are timestamped by their later arrival, and service durations by their completion. In the generic notation below, $\xi_i$ denotes a retained interarrival interval for the arrival process or a retained observed completed-job execution time for the service process. The screen takes into account only samples within the last $T_M>0$ seconds. Positive integers $n_{\min}\le n_{\max}$ set the minimum required sample count and cap the sample count during bursts. For a recent series $\xi_1,\ldots,\xi_n$, the screen first requires $n_{\min}\le n\le n_{\max}$. It then computes the coefficient of variation
\begin{equation}
 \mathrm{CV}=\frac{\sqrt{n^{-1}\sum_{i=1}^{n}(\xi_i-\bar{\xi})^2}}{\bar{\xi}}
\label{eq:screen-cv}
\end{equation}
and accepts the variability check when
\begin{equation}
 |\mathrm{CV}-1|\le\delta_{cv},
\label{eq:screen-cv-rule}
\end{equation}
where $\delta_{cv}\ge0$ is the tolerance around the exponential-distribution value $\mathrm{CV}=1$ \cite{gross2011}. 

Next, define the two lagged subsequences
\begin{equation}
\xi^{-}
=
(\xi_1,\ldots,\xi_{n-1}),
\qquad
\xi^{+}
=
(\xi_2,\ldots,\xi_n).
\end{equation}
The lag-one
autocorrelation is estimated as their sample Pearson correlation~\cite{pearson}:
\begin{equation}
\widehat{\rho}_1
=
\frac{
\widehat{\operatorname{Cov}}
\left(\xi^{-},\xi^{+}\right)
}{
\sigma_{\xi^{-}}\sigma_{\xi^{+}}
}.
\label{eq:screen-acf-value}
\end{equation}
Here, $\widehat{\rho}_1$ must satisfy
\begin{equation}
 |\widehat{\rho}_1|\le\delta_{\rho},
\label{eq:screen-acf}
\end{equation}
where $\delta_{\rho}\ge0$ is the maximum allowed absolute lag-one autocorrelation. This is motivated by the fact that, for an independent and identically distributed sequence of exponentially distributed observations, the population lag-one autocorrelation is zero, while the sample lag-one autocorrelation converges to zero as the sample size increases \cite{bartlett1946theoretical}. 

Finally, let $D_n^{\mathrm{KS}}$ be the two-sided Kolmogorov--Smirnov distance between the empirical cumulative distribution function and the fitted exponential CDF $1-\exp(-\xi/\bar{\xi})$ \cite{lilliefors1969}. The exponential-fit rule is
\begin{equation}
 D_n^{\mathrm{KS}}\le\frac{\kappa_m\kappa_0}{\sqrt{n}},
\label{eq:screen-ks}
\end{equation}
where $\kappa_0>0$ is the base acceptance scale and $\kappa_m>0$ is a configurable multiplier.

\begin{definition}[Series and workload acceptance]
\label{def:markov-acceptance}
A recent series is \emph{accepted} when it has at least $n_{\min}$ samples and satisfies every enabled condition in~\crefrange{eq:screen-cv-rule}{eq:screen-ks}. Let $A_A(t)$ and $A_S(t)$ denote arrival- and service-series acceptance. If service samples are required, then
\begin{equation}
A(t)=A_A(t)\land A_S(t).
\label{eq:workload-acceptance-required}
\end{equation}
If service samples are optional, then $A(t)=A_A(t)$.
\end{definition}

\begin{remark}[Operational interpretation]
The screen does not prove that the workload is Markovian, and it is not a standard hypothesis test. The data windows may overlap, and the thresholds are chosen settings. The controller checks the screen again at every scheduler call. An accepted result only means that the recent data are close enough to the exponential model for the controller to use it.
\end{remark}

\subsection{Policy and mode switching}

\begin{algorithm}[t]
\caption{SNF-ICON}
\label{alg:snf-icon}
\begin{algorithmic}[1]
\STATE Refresh node and queue state; consume new arrivals and completions. Update interarrival, service, runtime-ratio, and requested-node statistics.
\STATE Evaluate the recent-window screen and obtain $A(t)$.
\IF{$A(t)=0$}
    \STATE Execute baseline SNF+IPM immediate dispatch and future planning.
    \STATE Return (go to the end of the algorithm)
\ENDIF
\STATE Enter the ICON path.
\STATE Start every SNF-ordered job that fits immediately.
\STATE Predict active-job releases and build a sequential future-SNF plan for the remaining queue.
\IF{the remaining queue is empty and $t\in\mathcal{E}$}
    \IF{the arrival-recency gate is open}
        \STATE Compute the optimal number of warm-spare nodes and apply the selected plan.
    \ELSIF{the call contains a completion}
        \STATE Immediately switch off newly released, unreused nodes.
    \ENDIF
\ENDIF
\IF{the remaining queue is empty}
    \STATE Apply the hard idle-timeout policy.
\ENDIF
\STATE Record the mode, screen, gate, predictions, costs, and actions.
\end{algorithmic}
\end{algorithm}

\begin{definition}[Invocation, mode, and spare-planning point]
\label{def:decision-mode}
Let $\mathcal{V}$ denote the set of scheduler invocations, including the initial call, arrival and completion calls, and transition or timeout callbacks. The operating mode is
\begin{equation}
M(t)=
\begin{cases}
1 \qquad\text{(SNF+IPM fallback)}, & A(t)=0,\\
2 \qquad\text{(ICON)}, & A(t)=1.
\end{cases}
\label{eq:mode-selection}
\end{equation}
Let $\mathcal{E}\subset\mathcal{V}$ contain the initial call, arrival, and completion calls. A \emph{spare-planning point} is an invocation time $t\in\mathcal{E}$ at which no job remains in the queue after immediate dispatch. 
\end{definition}

At every invocation, newly appended monitor records are consumed before $A(t)$ is evaluated. The selected mode applies only to the current call. A rejected window therefore enters fallback only for the current call, and a later window may return to ICON mode. 

The overview of the SNF-ICON policy is shown in~\cref{alg:snf-icon}. It separates the handling of queued jobs from the handling of future jobs. Each invocation first updates the state and evaluates $A(t)$. Rejected windows delegate the complete invocation to SNF+IPM, whereas accepted windows use ICON-specific release prediction and future planning. Warm-spare optimization is considered only after all queued jobs have started execution, the invocation is a qualifying decision epoch, and the arrival-recency gate is open. Thus, jobs that have arrived always take precedence over future job preparation.

In cases where the arrival-recency gate is closed, the system remains in ICON mode because it suppresses warm spare node preparation but does not select the SNF+IPM fallback. It is important to note that $A(t)=1$ does not by itself trigger warm spare optimization. The ICON path must first leave the job queue empty by executing all jobs that have arrived, the invocation must be a qualifying decision epoch, and the arrival-recency gate must be open.


\subsection{Queued-job handling in the two operating modes}
\label{sec:visible-work}

Both operating modes use SNF scheduling for queued jobs. The main difference is how they predict node release times and plan wake-ups. In SNF+IPM fallback mode, the scheduler uses the baseline SNF+IPM procedure and then returns. In ICON mode, the controller first updates its arrival, service, and requested-node estimates, as described in~\cref{subsec:online-model}, and then builds the future plan.

In ICON mode, queued jobs are sorted using the SNF key in~\cref{eq:snf-key}. The scheduler starts each job that fits using the active-idle nodes available at time $t$. The scan stops at the first job that cannot fit, as described in~\cref{sec:system}.

For each job left in the queue, the controller estimates the earliest time when enough nodes will be available. It then selects the required nodes using the cost key defined below.

\begin{definition}[Sequential future-SNF plan]
\label{def:future-plan}
Let the jobs left after immediate dispatch be $j_1,\ldots,j_q$ in SNF order. The scheduler plans their executions one at a time. Before planning job $j_k$, let $\widehat{t}^{\mathrm{rel}}_{i,k}$ be the time when node $i$ is expected to become available, including assignments already planned for earlier jobs. If fewer than $r_{j_k}$ nodes can be included in the plan, planning stops. Let $b_k$ be the \emph{planning barrier}, initialized as $b_1=t$. For each job, let $\widehat{t}^{\mathrm{rel}}_{(r_{j_k}),k}$ be the $r_{j_k}$th earliest estimated node-release time. The planned execution start time for job $j_k$ is
\begin{equation}
 s_{j_k}=\max\left\{b_k,
 \widehat{t}^{\mathrm{rel}}_{(r_{j_k}),k}\right\}.
\label{eq:planned-start}
\end{equation}
Among the nodes expected to be available by $s_{j_k}$, the scheduler selects the $r_{j_k}$ nodes with the lowest estimated non-compute energy before the planned start, followed by node-state priority, remaining idle-timeout priority, and node identifier $i$. Node states are prioritized in the following order: idle, computing, switching on, and sleeping or switching off. For an idle node, the timeout priority is defined as the negative of its remaining idle timeout. Consequently, a longer remaining timeout yields a lower priority value and is selected first, leaving nodes closer to timeout available for power-off. Let the selected node set be $B_{j_k}$. Next, the planned (estimated) finish time for job $j_k$ is
\begin{equation}
 f_{j_k}=s_{j_k}+\widehat{D}_{j_k},
\label{eq:planned-finish}
\end{equation}
where $\widehat{D}_{j_k}$ is the predicted execution time of the queued job (see Appendix A
of the supplemental materials). A selected sleeping node $i$ is scheduled to wake at
\begin{equation}
 u_{i,j_k}=s_{j_k}-L_i^{\uparrow}.
\label{eq:planned-node-wake-pairs}
\end{equation}

After job $j_k$ is added to the plan, the estimated next-available time of every node in $B_{j_k}$ is changed to $f_{j_k}$. The next job uses $b_{k+1}=s_{j_k}$, so planned start times do not move backward.
\end{definition}

Different queued jobs may have the same planned start when they use different nodes. However, a later job in SNF order is never planned to start before an earlier one.

If a sleeping node is selected for more than one planned job, the scheduler keeps its earliest wake time:
\begin{equation}
 t_i^{\uparrow}=\min_{j:i\in B_j}\left(s_j-L_i^{\uparrow}\right),
 \qquad i\in\mathcal{S}(t).
\label{eq:planned-wake}
\end{equation}
If this time has already arrived, the node is switched on immediately. Otherwise, a callback is scheduled for that time. If the node is currently switching off, it is switched on after reaching the sleeping state, as long as this still meets the planned start. The sequential future-SNF plan is illustrated in~\cref{fig:seq-snf-plan}.

\begin{figure*}[t]
    \centering
    \includegraphics[width=0.98\linewidth]{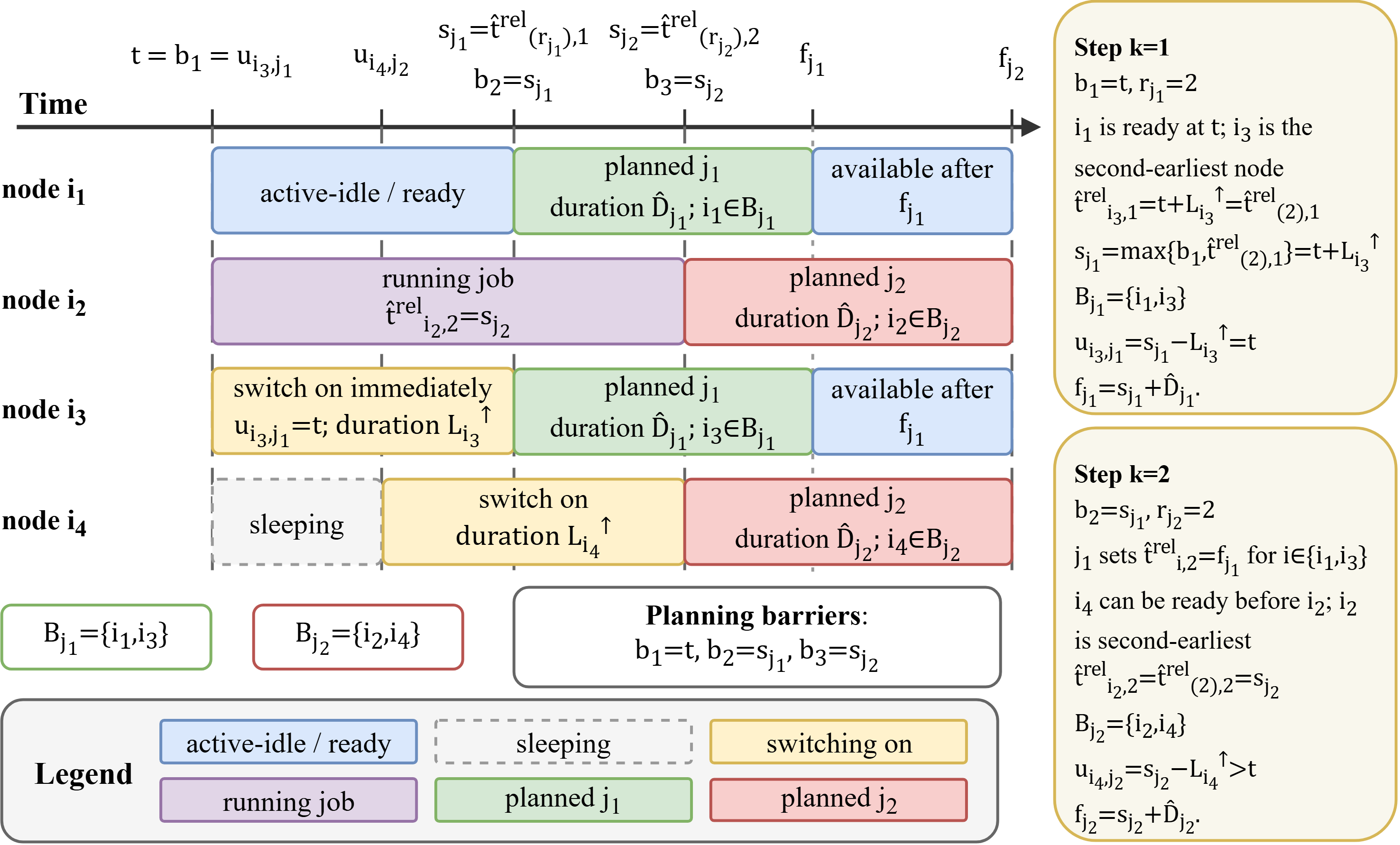}
    \caption{Sequential future-SNF plan for the remaining queued jobs, shown as a Gantt chart after immediate dispatch. The current scheduler time is $t$, and the remaining queue is ordered as $j_1 \prec_{\mathrm{SNF}} j_2$, with $r_{j_1}=r_{j_2}=2$. Node $i_3$ is instructed to wake immediately at $u_{i_3,j_1}=b_1=t$, causing job $j_1$ to start after the wake-up transition at $s_{j_1}=t+L_{i_3}^{\uparrow}$. After $j_1$ is added to the plan, the release map is updated and the next planning barrier becomes $b_2=s_{j_1}$. Job $j_2$ is then planned to start at $s_{j_2}$, with node $i_4$ scheduled to wake at $u_{i_4,j_2}=s_{j_2}-L_{i_4}^{\uparrow}$. The selected node sets are $B_{j_1}=\{i_1,i_3\}$ and $B_{j_2}=\{i_2,i_4\}$, and the planned finish time of each job is $f_{j_k}=s_{j_k}+\widehat{D}_{j_k}$.}
    \label{fig:seq-snf-plan}
\end{figure*}

Meanwhile, the SNF+IPM fallback mode plans queued jobs and their wake times using the baseline SNF+IPM estimates and procedures~\cite{prasasta2026}. It therefore keeps normal queued-job handling while leaving out ICON-specific prediction and warm-spare planning.

\subsection{Adaptive arrival, service, and resource estimates}
\label{subsec:online-model}

Let $\Delta_A$ be the newest observed interarrival interval (the newest observation of the arrival-screen series denoted generically by $\xi_i$ above), let $m_A$ be the current exponential moving average (EMA) of interarrival time, and let $\rho_A\in(0,1]$ be the arrival smoothing factor \cite{roberts2000}. The updated EMA is
\begin{equation}
 m_A^{+}=\rho_A \Delta_A+(1-\rho_A)m_A,
\qquad
\widehat{\lambda}(t)=\frac{1}{\max\{\epsilon,m_A^{+}\}}.
\label{eq:arrival-ema}
\end{equation}
After the update, $m_A$ is replaced by $m_A^{+}$. Here, $\widehat{\lambda}(t)$ is the estimated arrival rate and $\epsilon>0$ is a small numerical constant that prevents division by zero. Zero gaps between arrivals at the same simulation timestamp are retained. At initialization, $m_A=1/\lambda_0$, where $\lambda_0>0$ is the configured initial arrival rate.

Let $m_S$ be the EMA of observed completed-job durations. Before a usable completion exists, $m_S^{\mathrm{req}}$ stores an EMA of requested-duration proxies. Both service EMAs use smoothing factor $\rho_S\in(0,1]$. The estimated service rate is
\begin{equation}
\widehat{\mu}(t)=
\begin{cases}
1/\max\{\epsilon,m_S\}, & m_S\text{ is available},\\
1/\max\{\epsilon,m_S^{\mathrm{req}}\}, & \text{otherwise},\\
0, & \text{if neither is available}. 
\end{cases}
\label{eq:service-rate}
\end{equation}


For each possible requested node count $r\in\{1,\ldots,N\}$, let $w_r\ge0$ be its adaptive weight. When the newest job requests $r^{\prime}$ nodes, the weights are updated using the resource smoothing factor $\rho_R\in(0,1]$:
\begin{align}
 w_r^{+}&=(1-\rho_R)w_r+\rho_R\mathbf{1}\{r=r^{\prime}\}, \nonumber
\\
\widehat{p}_r(t)&=\frac{w_r^{+}}{\sum_{k=1}^{N} w_k^{+}}.
\label{eq:resource-ema}
\end{align}
After the update, each $w_r$ is replaced by $w_r^{+}$. Thus, $\widehat{p}_r(t)$ is the estimated probability that the next job requests $r$ nodes. 
The weights are initialized to $w_{r_0}=1$ and $w_r=0$ for all $r\neq r_0$, where $r_0$ is the requested node count for the first job in the workload. Hence, the normalizing denominator is positive. Weights below the numerical tolerance are removed and the survivors are renormalized; if numerical pruning removes every weight, the implementation resets to $w_1=1$. Before any request is observed, this gives $\widehat{p}_{r_0}(t)=1$.

\subsection{Predicting the next event and time window}

For each active job $j$, let $\widehat{R}_j(t)$ denote the predicted remaining runtime of a running job (derived in Appendix A of the supplemental materials).
These job-specific estimates are used both in the predictive release map and in the completion-rate approximation below.

\begin{assumption}[One-step competing-event approximation]
\label{ass:event-race}
At an ICON decision, the time to the next arrival is modeled as exponential with rate $\widehat{\lambda}(t)$, the time to the next completion is approximated as exponential with rate $\widehat{\lambda}_C(t)$, and the two times are treated as independent from each other.
\end{assumption}
For independent exponential clocks, the minimum is exponential with a rate equal to the sum of the component rates, and the probability that a component clock occurs first is proportional to its rate \cite{norris1997,gross2011}.

In ICON mode, let $n_{\mathrm{act}}(t)$ be the number of distinct active jobs and define the set of usable positive completion predictions as
\begin{equation}
 \mathcal{J}_{\mathrm{pred}}(t)=
 \left\{j:\ 0<\widehat{R}_j(t)<\infty\right\}.
\end{equation}
The earliest predicted remaining runtime is
\begin{equation}
 d_C(t)=\min_{j\in\mathcal{J}_{\mathrm{pred}}(t)}\widehat{R}_j(t),
 \qquad \min\varnothing:=\infty.
\label{eq:next-completion-delay}
\end{equation}
The estimated completion-event rate is
\begin{equation}
 \widehat{\lambda}_C(t)=
 \begin{cases}
 1/d_C(t), & 0<d_C(t)<\infty,\\
 n_{\mathrm{act}}(t)\widehat{\mu}(t), & \text{otherwise}.
 \end{cases}
\label{eq:completion-hazard}
\end{equation}
Thus, the first branch uses the nearest job-specific completion prediction, while the second uses the aggregate service-rate approximation. The total next-event rate $\widehat{\lambda}_E(t)$ and the probability $\pi_A(t)$ that the next event is an arrival are
\begin{align}
 \widehat{\lambda}_E(t)&=\widehat{\lambda}(t)+\widehat{\lambda}_C(t),\nonumber\\
 \pi_A(t)&=\begin{cases}
 \widehat{\lambda}(t)/\widehat{\lambda}_E(t), & \widehat{\lambda}_E(t)>0,\\
 0, & \widehat{\lambda}_E(t)=0.
 \end{cases}
\label{eq:race}
\end{align}

Let $d_{\mathrm{to}}\ge0$ be the time remaining until the earliest future idle-node timeout, or $H_{\max}$ when no such timeout exists. Here, $H_{\max}>0$ is the configured maximum prediction horizon. The capped horizon is
\begin{equation}
 H_c=\min\{H_{\max},d_{\mathrm{to}}\}.
\label{eq:horizon-cap}
\end{equation}
\begin{lemma}[Timeout-capped exponential horizon]
\label{lem:truncated-horizon}
Let $T$ denote the modeled time to the next arrival-or-completion event. If $T\sim\mathrm{Exp}(\widehat{\lambda}_E(t))$ with $\widehat{\lambda}_E(t)>0$ and $H_c\ge0$, then
\begin{equation}
 H=\E[\min(T,H_c)]
 =\frac{1-e^{-\widehat{\lambda}_E(t) H_c}}{\widehat{\lambda}_E(t)}.
\label{eq:truncated-horizon}
\end{equation}
\end{lemma}

\begin{proof}
Using the tail-integral identity for a nonnegative random variable,
\begin{equation}
\E[\min(T,H_c)]
=\int_0^{H_c}\Prob(T>u)\,\mathrm{d}u
=\int_0^{H_c}e^{-\widehat{\lambda}_E(t) u}\,\mathrm{d}u,
\end{equation}
which evaluates to~\cref{eq:truncated-horizon}.
\end{proof}


\subsection{Valid spare-node targets}

The integer target $x$ is the number of non-computing nodes that the plan prepares as warm spare nodes.

Let $\mathcal{I}_e(t)$ be the set of idle nodes whose timeout has expired:
\begin{equation}
\mathcal{I}_e(t)
=
\left\{
i\in\mathcal{I}(t):
\text{node $i$'s timeout has expired}
\right\}.
\label{eq:expired-idle}
\end{equation}
These nodes are switched off in every candidate plan. Next, among the remaining idle nodes, let $\mathcal{I}_f(t)$ contain the nodes that cannot currently be switched off:
\begin{equation}
\mathcal{I}_f(t)
=
\left\{
i\in\mathcal{I}(t)\setminus\mathcal{I}_e(t):
\text{node $i$ cannot switch off}
\right\}.
\label{eq:forced-warm-idle}
\end{equation}
When the timeout policy is enabled, this includes idle nodes whose timeout has not yet expired. Hence, these nodes must remain warm. The other idle nodes may remain warm or be switched off depending on the selected target $x$.

The number of nodes that can be used as warm spare nodes is
\begin{equation}
x_{\mathrm{real}}
=
|\mathcal{U}(t)|
+
|\mathcal{I}(t)\setminus\mathcal{I}_e(t)|
+
|\mathcal{S}(t)|.
\label{eq:realizable-spares}
\end{equation}

Let $S_{\min},S_{\max}\in\{0,\ldots,N\}$ be the configured minimum and maximum spare-node targets, where $S_{\min}\le S_{\max}$. Nodes that are already switching on and idle nodes in $\mathcal{I}_f(t)$ must remain part of every candidate plan. Therefore, the feasible target range is
\begin{align}
x_{\min}
&=
\max\!\left\{
|\mathcal{U}(t)|+|\mathcal{I}_f(t)|,
\min\{S_{\min},x_{\mathrm{real}}\}
\right\},
\nonumber\\
x_{\max}
&=
\max\!\left\{
x_{\min},
\min\{S_{\max},x_{\mathrm{real}}\}
\right\}.
\label{eq:target-range}
\end{align}

The value $S_{\max}$ is a preferred limit rather than a strict limit because the number of nodes that must remain warm may already exceed it. These bounds satisfy
\begin{equation}
0\le x_{\min}\le x_{\max}\le x_{\mathrm{real}},
\end{equation}
so the candidate set $\{x_{\min},\ldots,x_{\max}\}$ is always nonempty.

For each candidate target $x$, the plan keeps all nodes that are already switching on and all idle nodes in $\mathcal{I}_f(t)$. It then keeps additional idle nodes or wakes sleeping nodes until the target is reached. Nodes in $\mathcal{I}_e(t)$ are switched off, and other unused idle nodes may also be switched off.

\subsection{Power-induced waiting estimate}

Power-induced waiting is the delay caused by non-computing nodes that are not ready to run a newly arriving job. Waiting caused by nodes that are still running jobs is not included because it is the same for every spare-node target and does not affect which target is selected.

Let
\begin{equation}
K=N-|\mathcal{C}(t)|
\end{equation}
be the number of non-computing nodes at time $t$. For a candidate target $x$, let $d_i(x)$ be the time until non-computing node $i$ becomes ready under that plan.
An idle node that remains warm has zero delay. A node that is already switching on, or is instructed to wake immediately, uses the expected remaining transition time $\phi(\cdot;\widehat{\lambda}_E(t))$ from (59) in Appendix B of the supplemental materials.
A node that remains asleep uses its full wake-up time. A node that is switching off includes both its expected remaining switch-off time and its wake-up time.

For $K>0$, sort the readiness delays as
\begin{equation}
 d_{(1)}(x)\le\cdots\le d_{(K)}(x),
\end{equation}
where $d_{(r)}(x)$ is the estimated time by which the $r$-th non-computing node becomes ready. The expected request-weighted readiness delay is
\begin{equation}
 \bar d(x)=\sum_{r=1}^{N}\widehat{p}_r(t)\,
 d_{(\min\{r,K\})}(x),
\end{equation}
where $\widehat{p}_r(t)$ is the estimated probability that the next job requests $r$ nodes. Next, the power-induced waiting estimate is
\begin{equation}
 W(x)=
 \begin{cases}
 \pi_A(t)\bar d(x), & K>0,\ \pi_A(t)>0,\\
 0, & K=0\ \text{or}\ \pi_A(t)=0.
 \end{cases}
\label{eq:expected-wait}
\end{equation}
Here, $\pi_A(t)$ is the probability that the next event is an arrival. When $r>K$, the index is limited to $K$ because any additional waiting is caused by nodes that are still running jobs and is the same for every candidate plan.
The estimate $W(x)$ uses the expected remaining transition times before sorting the node delays. It is therefore an approximation of the expected waiting time rather than its exact value.

\begin{figure*}[t]
    \centering
    \includegraphics[width=0.98\linewidth]{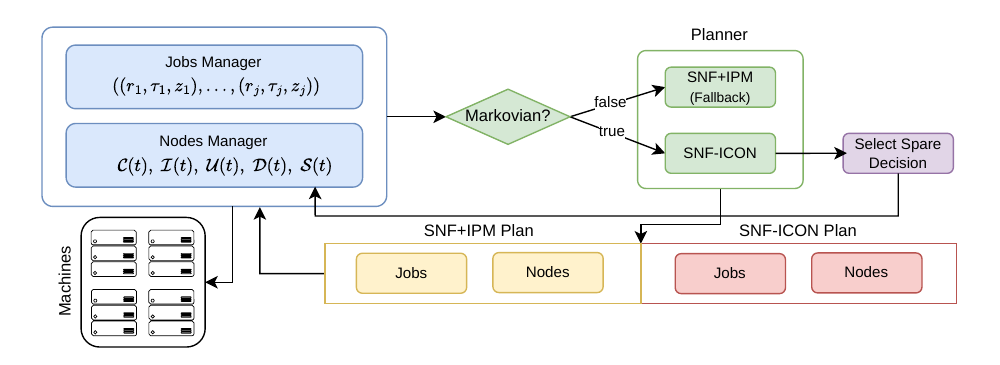}
    \caption{Implementation architecture in SPARS~\cite{amrizal2026}.}
    \label{fig:architecture}
\end{figure*}
\subsection{Estimated energy and spare-node selection}

For each node, the controller estimates the energy used by sleeping, switching, and idle states over the interval $[t,t+H]$. Here, $H$ is the mean capped time until the next event from~\cref{eq:truncated-horizon}. Compute energy is not included because it is required by running jobs and is not controlled by the spare-node target.
For example, an idle node that remains warm uses $P_i^{\mathrm{idle}}H$, where $P_i^{\mathrm{idle}}$ is the idle power of node $i$. A node with $a_i$ units of switch-on time remaining uses
\begin{equation}
 E_i^{\uparrow}(H)=
 P_i^{\uparrow}\min(a_i,H)+
 P_i^{\mathrm{idle}}\pos{H-a_i},
\label{eq:energy-up}
\end{equation}
where $P_i^{\uparrow}$ is its switch-on power. Energy for waking, sleeping, and switching off is calculated in the same way. Let $E_i^{\mathrm{nc}}(x,H)$ be the non-compute energy of node $i$ under candidate target $x$. The estimated total non-compute energy of the whole system is
\begin{equation}
 E(x)=\sum_{i=1}^{N}E_i^{\mathrm{nc}}(x,H).
\label{eq:expected-energy}
\end{equation}
This value is an estimate based on the mean capped horizon $H$, rather than the exact expected energy over a random time period. Next, the set of candidate spare-node targets is
\begin{equation}
 \mathcal{X}(t)=\{x_{\min},\ldots,x_{\max}\}.
\end{equation}
The total candidate cost and selected target are
\begin{align}
 J(x)&=\alpha W(x)+\beta E(x),\nonumber\\
 x^{\star}&=\operatorname*{arg\,lex\,min}_{x\in\mathcal{X}(t)}
 \bigl(J(x),W(x),E(x),x\bigr).
\label{eq:objective}
\end{align}
where $\alpha,\beta\ge0$ are the weights for waiting time and wasted energy. The lexicographic minimization first minimizes $J(x)$, then $W(x)$, then $E(x)$, and finally $x$; hence $x^{\star}$ is uniquely and deterministically selected.

\subsection{Arrival-recency gate, completion shutdown, and idle timeout}
\label{subsec:gate}

\begin{definition}[Arrival-recency gate]
\label{def:arrival-gate}
Before the first recorded arrival, the arrival-recency gate is closed. After that, let $t_A$ be the time of the latest recorded arrival and let $T_A^{\max}\ge0$ be the configured recency window. The gate is \emph{open} at time $t$ when
\begin{equation}
 t-t_A\le T_A^{\max}.
\label{eq:gate}
\end{equation}
\end{definition}



The arrival-recency gate is separate from the fallback rule and is checked only in ICON mode after all queued jobs have been handled. When open, it allows warm-spare planning for future jobs; when closed, ICON remains active but does not prepare spare nodes. Unlike the Markovianity screen, which checks whether recent workload data fit the model, the gate checks whether the latest arrival is recent enough to make that data useful.

If a job finishes while the gate is closed, newly released nodes are switched off unless they are needed by a queued or planned job. Other idle nodes follow the normal timeout policy: each is switched off after remaining unused for $\Delta_t$, and the next timeout also limits the energy horizon in~\cref{eq:horizon-cap}.

\section{Implementation in SPARS}
\label{sec:implementation}

The implementation for this study extends the SNF+IPM baseline within the SPARS discrete-event simulation framework~\cite{amrizal2026}. The source code, experimental configurations, and visualization scripts used in this study are publicly available in the SPARS-ICON repository \cite{github} to facilitate reproducibility. \cref{fig:architecture} summarizes the resulting control flow architecture. At each scheduler invocation, the jobs and nodes managers provide the job and node-state information to the Markovianity screen. If the Markovianity check is rejected, the invocation follows the baseline SNF+IPM planner. Otherwise, it follows the SNF-ICON planner and evaluates the warm-spare target. The resulting job schedule and power-state actions are returned to the SPARS managers for execution by the simulated machines.


The rest of this section describes the implementation details in SPARS.

\subsection{Event detection and diagnostics}

At every scheduler invocation, the implementation reevaluates the screen and records \texttt{policy\_mode} as either \texttt{stochastic\_spare\_capacity} or \texttt{snf\_fallback}. The forced-ICON no-fallback variant is exposed only as an experimental ablation. The log separately records whether the call is a spare-planning point and whether the arrival-recency gate is open. This separation prevents a closed gate or a nonempty queue from being miscounted as fallback. Diagnostics include rejection reasons, sample counts, means, coefficients of variation, lag-one correlations, KS distances and thresholds, arrival and service rates, runtime-ratio statistics, active-job remaining-time predictions, requested-node PMF, future planned jobs, immediate wakes, future wake callbacks, spare-target candidate costs, timeout actions, and nodes attributed to a completion.

\subsection{Energy-aware node choice and release maps}

The SPARS implementation realizes the sequential future-SNF planner defined in~\cref{sec:visible-work} using the simulator's node lists and release maps. In ICON mode, active compute segments are replaced by job-specific predicted finish times, and planned jobs append synthetic compute segments. Fallback mode instead invokes the corresponding SNF+IPM planning routines. Both modes may schedule wakes for queued jobs, but only ICON performs warm spare node planning.

\subsection{Default parameters}


Unless otherwise stated, the experiments use the default parameter values shown in~\cref{tab:defaults}. They are starting points rather than universally optimal values.

\begin{table}[t]
\caption{Implemented default parameters.}
\label{tab:defaults}
\centering
\begin{tabular}{lll}
\toprule
Parameter & Symbol & Default \\
\midrule
Waiting weight & $\alpha$ & $5000$ \\
Energy weight & $\beta$ & $1$ \\
Initial arrival rate & $\lambda_0$ & $1/3600\ \mathrm{s}^{-1}$ \\
Arrival EMA factor & $\rho_A$ & $0.2$ \\
Service/ratio EMA factor & $\rho_S$ & $0.1$ \\
Resource EMA factor & $\rho_R$ & $0.1$ \\
Log-ratio sigma floor & $\sigma_{\min}$ & $0.25$ \\
Minimum spare nodes & $S_{\min}$ & $1$ \\
Maximum spare nodes & $S_{\max}$ & platform size \\
Maximum horizon & $H_{\max}$ & $24$ h \\
Arrival-recency window & $T_A^{\max}$ & $12$ h \\
Markovianity fallback & -- & enabled \\
Screen time window & $T_M$ & $2$ h \\
Minimum screen samples & $n_{\min}$ & $3$ \\
Maximum screen samples & $n_{\max}$ & $1000$ \\
CV tolerance & $\delta_{cv}$ & $0.45$ \\
Maximum $|\widehat{\rho}_1|$ & $\delta_{\rho}$ & $0.65$ \\
KS multiplier & $\kappa_m$ & $1.25$ \\
KS scale parameter & $\kappa_0$ & $3.0$ \\
Require service samples & -- & enabled \\
Stale-completion shutdown & -- & enabled \\
Respect idle timeout & -- & enabled \\
Timeout duration & $\Delta_t$ & 2 h \\
\bottomrule
\end{tabular}%
\end{table}

\subsection{Computational complexity}

Let $q$ be the number of queued jobs, $N$ the number of nodes, $n\le n_{\max}$ the number of samples in each screen window, and $n_{\mathrm{act}}$ the number of active jobs. For each required series, the screen computes the summary statistics and lag-one autocorrelation in $\mathcal{O}(n)$ time. Computing the KS distance requires sorting the samples, giving $\mathcal{O}(n\log n)$ time per series. Predicting the remaining times of active jobs and updating their node release times takes $\mathcal{O}(n_{\mathrm{act}}+N)$ time, assuming each active job and node is visited once. This does not include the cost of copying release maps, which depends on how they are stored.

If the queue is not already kept in SNF order, sorting it takes $\mathcal{O}(q\log q)$ time. For each job in the future plan, the controller sorts at most $N$ node release times and at most $N$ node-selection keys. Each planning step therefore takes $\mathcal{O}(N\log N)$ time, and planning at most $q$ queued jobs takes $\mathcal{O}(qN\log N).$
Building the plan and its wake-up times requires storage proportional to the total number of assigned job--node pairs, $\sum_j r_j$.

The spare-node optimizer considers at most $N+1$ targets. For each target, building the node plan and estimating its energy take $\mathcal{O}(N)$ time, while sorting the node readiness times takes $\mathcal{O}(N\log N)$. The complete spare-node optimization therefore takes $\mathcal{O}(N^2\log N).$
If the candidate targets are evaluated one at a time, the additional working storage is $\mathcal{O}\left(n+N+\sum_j r_j\right)$, excluding the waiting queue, release-map copies, logs, and other simulator state.

These are analytical bounds based on the described algorithm. The experiments do not measure scheduler execution time, so the actual decision time must be checked through implementation profiling.

\begin{figure*}[!t]
    \centering
    \includegraphics[width=0.98\linewidth]{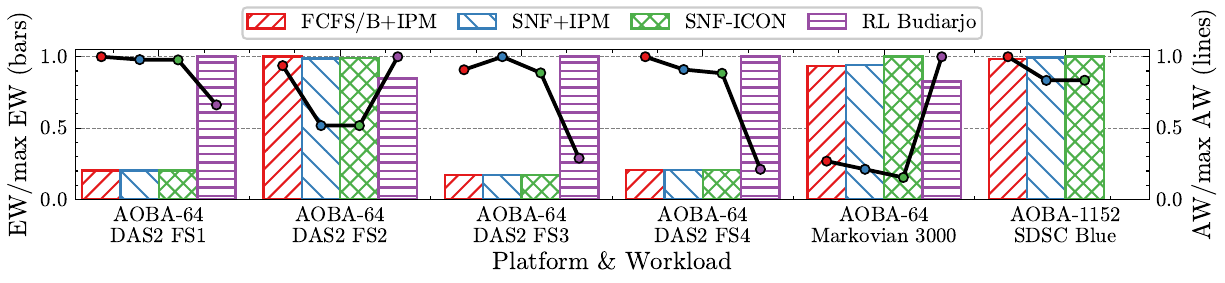}
    \caption{Base-configuration comparison of energy waste and average waiting, respectively denoted as EW and AW. RL Budiarjo uses a separately trained policy for each dataset; results are reported only for datasets with an available RL experiment.}
    \label{fig:base-barplot}
\end{figure*}

\section{Experimental Methodology}
\label{sec:experiments}

\subsection{Evaluation scope}

The evaluation focuses on average queue waiting time and total wasted energy as the primary outcome metrics, since they directly represent the scheduling-delay and energy-efficiency objectives of the proposed method. The remaining measurements, tests, and sensitivity analyses are used to interpret the algorithm's behavior and explain the observed energy--delay trade-offs. Accordingly, the evaluation addresses three questions supported by the supplied results. First, how do the fixed SNF-ICON configuration's average queue waiting time and total wasted energy compare with those of FCFS/B+IPM, SNF+IPM, and the available RL controller across the six workload--platform cases? Second, does the Markovianity screen distinguish the generated exponential workload from the DAS2 traces, and how do disabling fallback or the arrival-recency gate affect the resulting operating point? Third, how sensitive is the energy--delay trade-off to the reward weight, Markovianity lookback horizon, and platform model, and do the aggregate results conceal a size-dependent waiting-time penalty?

The plotting scripts denote the complete proposed method as \textbf{SNF-ICON}. The \textbf{NG} (No-Gate) suffix disables the arrival-recency gate, the \textbf{NF} (No-Fallback) suffix disables Markovianity-based fallback and therefore forces the ICON path, and \textbf{NGNF} (No-Gate No-Fallback) disables both safeguards. These four variants are compared with \textbf{SNF+IPM}, which retains SNF scheduling but uses IPM as the power manager, and \textbf{FCFS/B+IPM}, which is the PSAS+IPM configuration from the cited study \cite{prasasta2026}.

\subsection{Workloads and platform configurations}

The base evaluation contains six workload--platform cases. Five workloads run on a 64-node AOBA-derived configuration: the first 3000 jobs from each of the DAS2 FS1--FS4 traces and a 3000-job generated Markovian workload. The DAS2 traces originate from a multi-cluster production workload described in prior characterization work and are distributed through the Parallel Workloads Archive \cite{li2004workload,feitelson2014}. The generated workload uses exponential interarrival and service assumptions.

For the RL Budiarjo baseline, the subsequent 1000 jobs from each corresponding DAS2 trace are used to train a separate policy for that dataset. Training follows the curriculum simple $\rightarrow$ real $\rightarrow$ complex, where the simple and complex workloads are generated from the characteristics of the corresponding real training segment. The sixth evaluation case maps the SDSC Blue workload to an AOBA-derived 1152-node configuration. The AOBA configurations follow prior digital-twin modeling of Supercomputer AOBA \cite{ohmura2022}.

The Markovianity-window sweep evaluates the five AOBA-64 workloads by varying $T_M$ from 1~h to 24~h, whereas the objective-weight sweep varies $\alpha$ from 10 to 20{,}000 while fixing $\beta=1$, with all other parameters fixed at their base values in both sweeps. Cross-platform sensitivity evaluates three representative workloads---DAS2 FS2, Markovian 3000, and SDSC Blue---on corresponding AOBA and Taurus models, yielding six workload--platform panels. DAS2 FS2 and Markovian 3000 use 64-node models, whereas SDSC Blue uses 1152-node models.

\section{Results and Discussion}
\label{sec:results}

\begin{figure}[h]
    \centering
    \includegraphics[width=\linewidth]{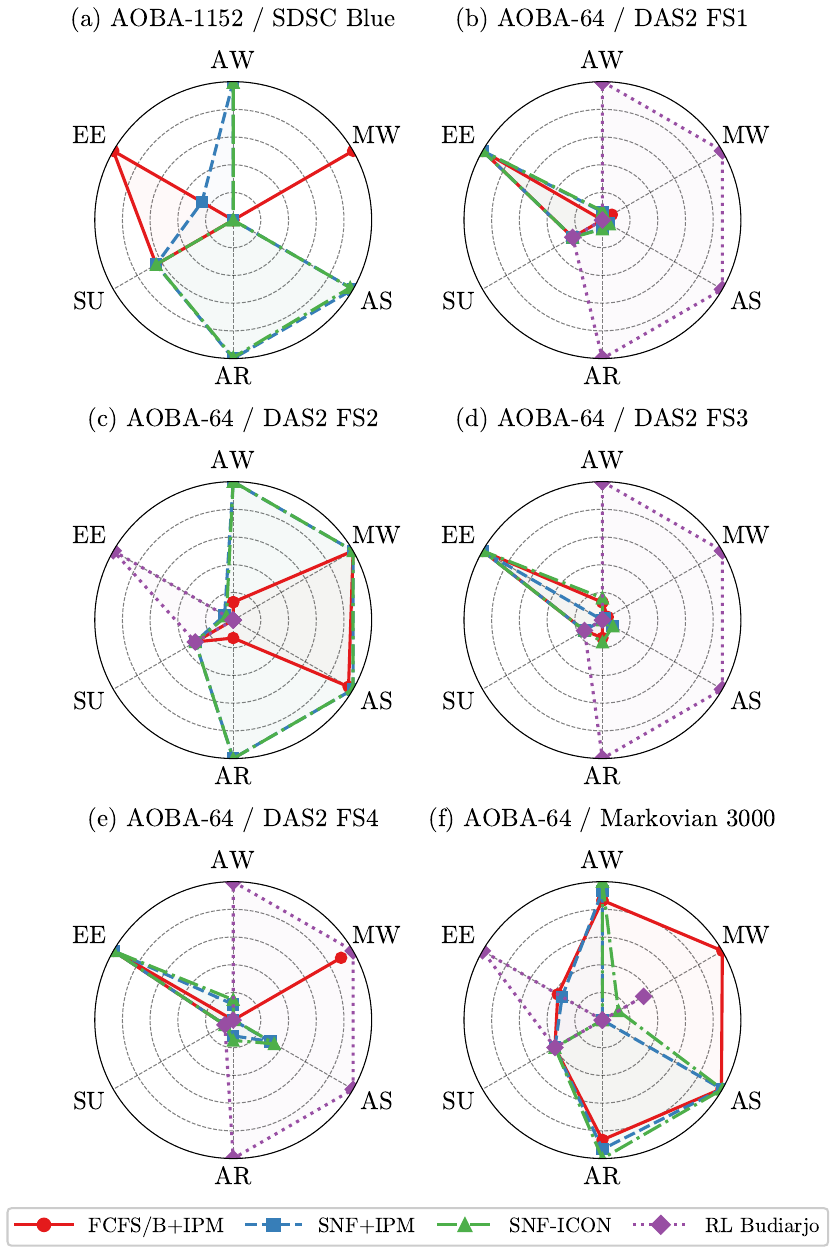}
    \caption{Normalized multi-metric profiles. AW: average wait, MW: maximum wait, AS: average slowdown, AR: average response, SU: system utilization, EE: energy efficiency. Delay and wasted-energy metrics are converted to efficiency scores so that larger radial values are better.}
    \label{fig:kiviat-base}
\end{figure}

\subsection{Base-configuration comparison across workloads}

\cref{fig:base-barplot} compares the fixed base configuration across the six workload-platform cases, with each metric normalized by the worst result within its panel. Relative to FCFS/B+IPM, SNF-ICON reduces mean waiting time in every panel: approximately by 2.2\% on DAS2 FS1, 44.8\% on DAS2 FS2, 2.4\% on DAS2 FS3, 11.6\% on DAS2 FS4, 42.4\% on the generated Markovian workload, and 16.4\% on SDSC Blue. Relative to SNF+IPM, the corresponding changes are improvements of approximately 0.2\%, 0.1\%, 11.3\%, 2.9\%, and 26.9\% in the first five cases, followed by a 0.15\% increase in waiting time on SDSC Blue. The base configuration therefore preserves the main waiting-time benefit of SNF and provides additional reductions on five of the six cases.

\begin{figure}[t]
    \centering
    \begin{subfigure}[t]{0.94\linewidth}
        \centering
        \includegraphics[width=\linewidth]{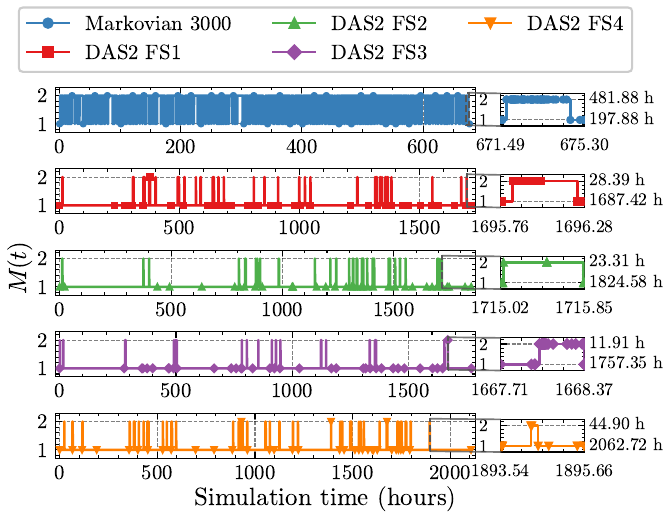}
        \caption{Mode sequence and accumulated duration.}
        \label{fig:markov-mode}
    \end{subfigure}

    \begin{subfigure}[t]{0.94\linewidth}
        \centering
        \includegraphics[width=\linewidth]{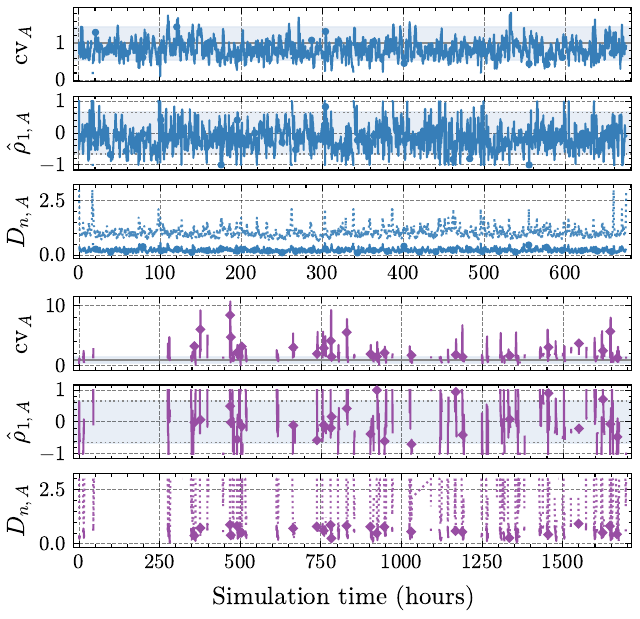}
        \caption{Interarrival diagnostics for Markovian 3000 and DAS2 FS3.}
        \label{fig:markov-conditions}
    \end{subfigure}

    \caption{Mode occupancy and recent-window Markovianity diagnostics. Mode 1 is SNF+IPM fallback and mode 2 is ICON.}
    \label{fig:markov-analysis}
\end{figure}

\begin{figure}[t]
    \centering
    \includegraphics[width=\linewidth]{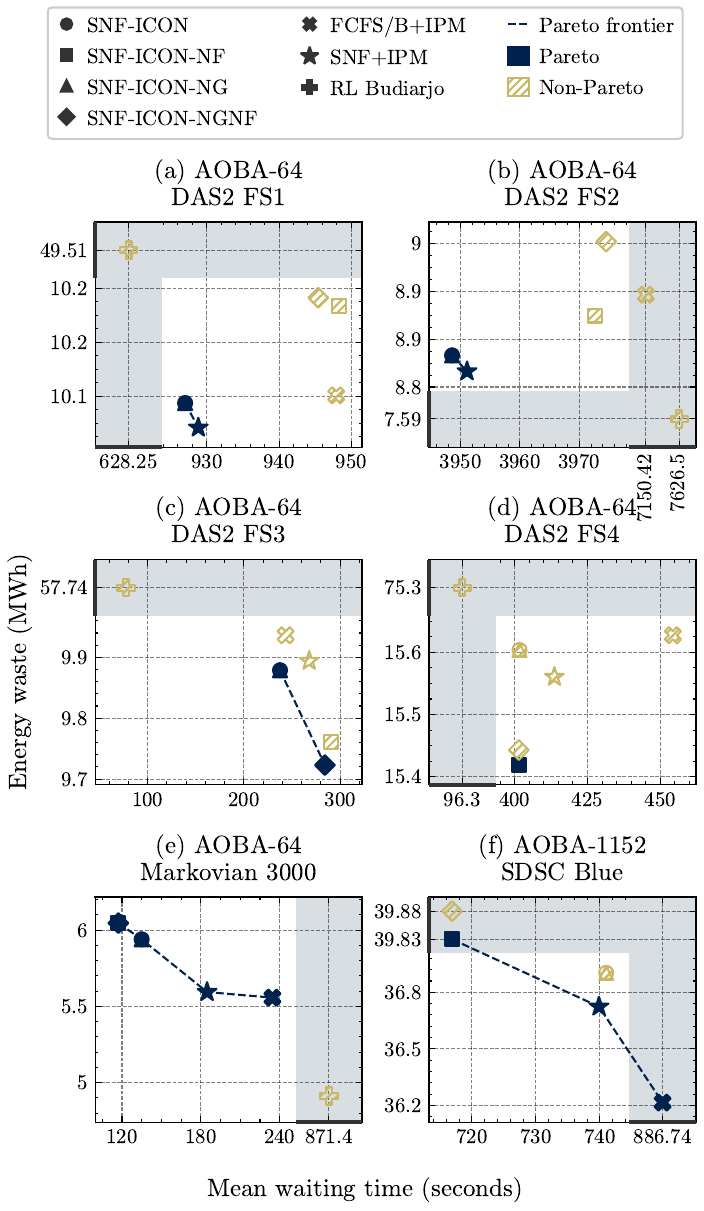}
    \caption{Waiting-time--energy comparison of the complete method and its variants across the six workload--platform cases.}
    \label{fig:variant-pareto}
\end{figure}

The energy differences among the heuristic methods are much smaller than the waiting-time differences in most cases. SNF-ICON consumes approximately 9.88~MWh on DAS2 FS3 and 15.6~MWh on DAS2 FS4. It is within about 1\% of at least one heuristic baseline in five cases. The exception is the generated Markovian workload, where its 5.94~MWh energy waste is approximately 6.2\% above SNF+IPM and 6.9\% above FCFS/B+IPM. Thus, the fixed configuration obtains its largest delay reductions without the multi-fold energy increase exhibited by the low-delay RL points, although the generated workload shows a measurable rather than negligible energy premium.

The dataset-specific RL policies produce more extreme trade-offs. On DAS2 FS1, FS3, and FS4, RL reduces mean waiting time by approximately 32\%, 67\%, and 76\%, respectively, but consumes approximately 4.9, 5.8, and 4.8 times the corresponding SNF-ICON energy. Conversely, on DAS2 FS2 and Markovian 3000, RL reduces energy consumption by approximately 14\% and 17\%, respectively, but incurs approximately 1.9 and 6.5 times the corresponding SNF-ICON mean waiting time. 
The available RL policy therefore supplies workload-specific extreme points.

\subsection{Multi-metric performance profile}

The Kiviat profiles in~\cref{fig:kiviat-base} confirm the trade-offs observed in~\cref{fig:base-barplot}. On DAS2 FS2 and Markovian 3000, SNF-ICON reaches or approaches the outer envelope on the delay spokes. On DAS2 FS1, FS3, and FS4, RL performs better on several delay spokes but has the weakest or nearly weakest energy-efficiency score. SNF-ICON and SNF+IPM are nearly identical on SDSC Blue, consistent with their 0.15\% waiting-time difference and small energy separation. Overall, SNF-ICON generally matches or improves the delay profile of SNF+IPM while remaining substantially more energy-efficient than the low-delay RL points. Because each spoke is normalized independently, the profiles indicate balanced performance rather than effect size.

\begin{figure*}[t]
    \centering
    \includegraphics[width=0.98\linewidth]{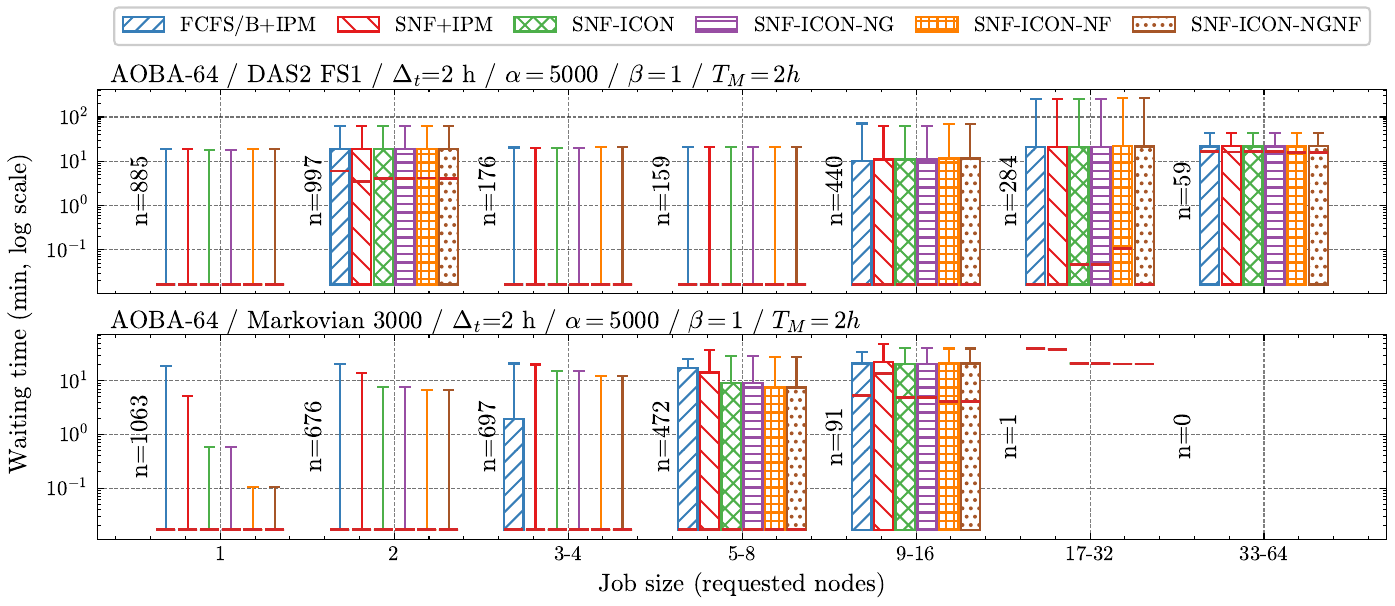}
    \caption{Waiting time stratified by requested node count. Boxes show the interquartile range, center lines show medians, and whiskers show the 5th--95th percentiles. Zero waits are displayed at one second on the logarithmic axis.}
    \label{fig:starvation-new}
\end{figure*}

\subsection{Why frequent fallback does not remove the advantage}

In~\cref{fig:markov-analysis}, mode~1 represents SNF+IPM fallback and mode~2 represents ICON. The values on the right side of each panel in~\cref{fig:markov-mode} show the total time spent in ICON mode at the top and fallback mode at the bottom. On the Markovian 3000 dataset, the policy spends 70.9\% of its simulated time in ICON mode. In contrast, DAS2 FS1--FS4 spend only 0.7--2.1\% of their time in ICON mode and therefore operate mainly in fallback.

The diagnostic panels in~\cref{fig:markov-conditions} explain this difference. Blank intervals indicate that too few recent samples were available, while the shaded regions and KS limits show the acceptance ranges. The DAS2 FS3 workload often enters fallback because arrivals and completions are sparse or because the observed statistics do not satisfy the workload screen.

Despite the frequent fallbacks, many fallback intervals occur during quiet periods with few job arrivals. The scheduling mode during these periods has less effect because only a small number of jobs are affected. ICON is more likely to become active during busier intervals, when enough samples are available and more jobs compete for resources. Its release prediction and warm-spare planning can therefore have a greater effect during these important periods.

This helps explain why frequent fallback does not remove the overall benefit of SNF-ICON. The results suggest that short ICON intervals during busier periods contribute to the observed waiting-time improvements, although the experiments do not directly measure the performance of each mode within individual intervals.

\subsection{Variant comparison and empirical Pareto frontier}

\cref{fig:variant-pareto} compares the complete method and its gate/fallback variants at the base parameters. Throughout the waiting-time--energy scatter plots, light-gray bands mark compressed portions of the axes containing isolated extreme points, thereby keeping the main data cluster visually distinguishable. The complete SNF-ICON point lies on the empirical Pareto frontier in four of the six panels, on DAS2 FS1, FS2, FS3, and Markovian 3000. When all proposed variants are considered, the SNF-ICON family contributes at least one nondominated point in every panel. On DAS2 FS1 and FS2, the complete method and SNF+IPM form closely spaced energy--delay trade-offs. On DAS2 FS3, the complete point supplies the low-delay endpoint, while NGNF supplies the lower-energy endpoint. On DAS2 FS4, the NF point dominates the complete point by approximately 0.06\% in waiting time and 0.95\% in energy waste. Several proposed points also lie on the Markovian 3000 frontier.

SDSC Blue is the exception to the frontier claim. There, SNF+IPM slightly dominates the complete configuration, reducing mean waiting time by approximately 0.15\% and wasted energy by 0.41\%. However, the NF configuration supplies another nondominated operating point, reducing waiting time by approximately 3.1\% relative to SNF+IPM while increasing wasted energy by approximately 8.6\%. 

Across the panels, NG usually overlaps or remains very close to the complete method, while NGNF remains close to NF. The visible change from disabling fallback is therefore larger than the visible change from disabling the arrival-recency gate at the base settings. The proposed family broadens the available energy--delay trade-off on most workloads, but no single variant is nondominated everywhere.

\subsection{Starvation and waiting time by requested job size}

\cref{fig:starvation-new} evaluates whether the aggregate advantage hides a size-dependent waiting-time penalty. The upper panels report DAS2 FS1 and the lower panels report Markovian 3000 at base configuration. Waiting-time distributions generally broaden as requested node count increases because larger rigid jobs require more nodes to become available simultaneously. In many small- and medium-size bins, SNF-ICON has lower or comparable medians to FCFS/B+IPM, and the 5th--95th percentile ranges do not show a consistent additional widening relative to FCFS/B.

The large-job evidence must be interpreted by sample count. The DAS2 FS1 workload contains 284 jobs requesting 17--32 nodes and 59 requesting 33--64 nodes. These bins have broad waiting-time ranges under all methods, but the proposed variants remain broadly comparable with FCFS/B+IPM and SNF+IPM. The Markovian 3000 workload contains only one 17--32-node job and no 33--64-node jobs, so those panels cannot support a large-job fairness conclusion. The distributions therefore show no obvious additional starvation penalty from SNF-ICON, but they do not establish starvation freedom. 

\begin{figure}[t]
    \centering
    \includegraphics[width=\linewidth]{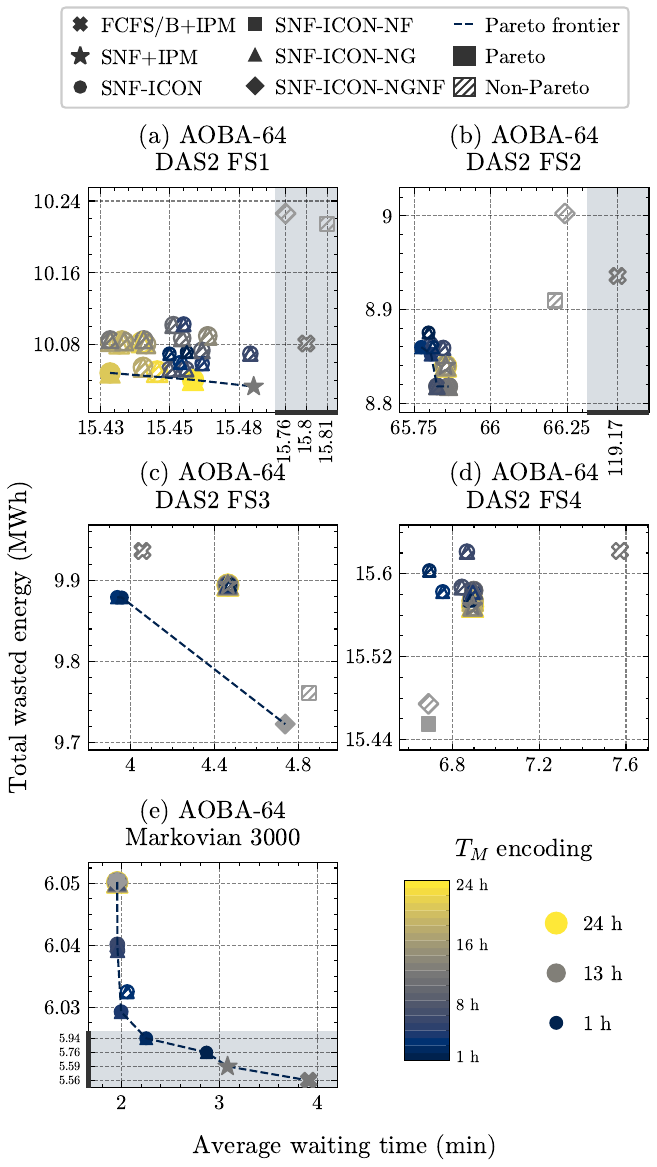}
        \caption{Markovianity lookback sweep.}
        \label{fig:param-ablation-tm}
\end{figure}
\begin{figure}[t]
    \centering
    \includegraphics[width=\linewidth]{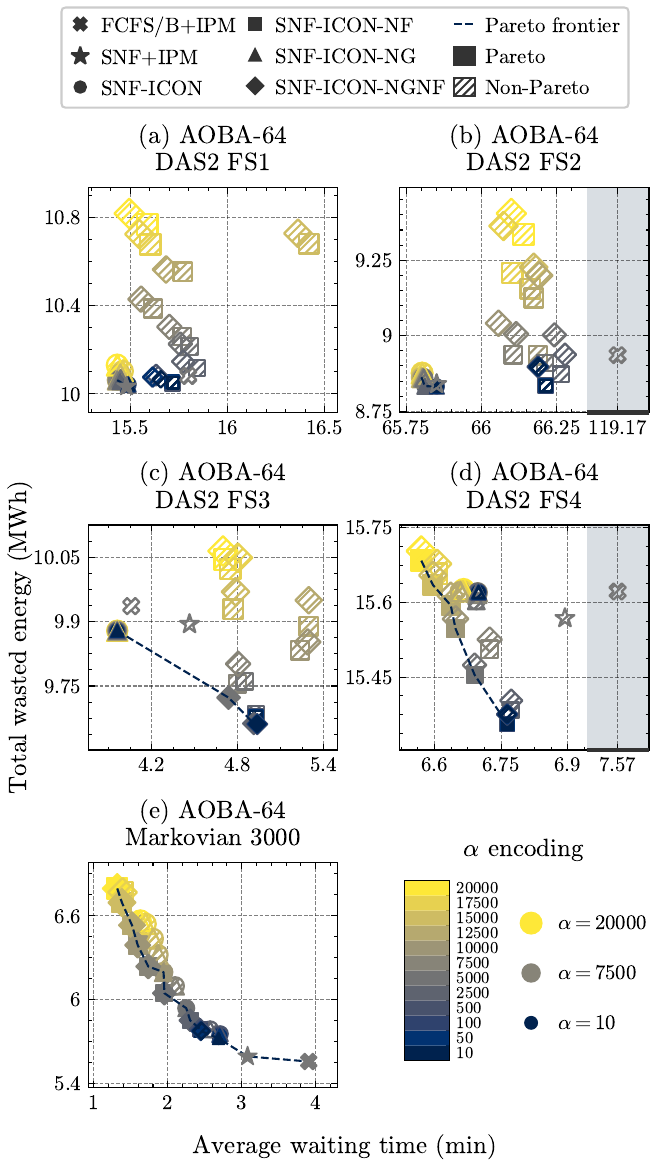}
    \caption{Reward-parameter sweep.}
    \label{fig:param-ablation-ab}
\end{figure}

\subsection{Parameter ablations}

\Cref{fig:param-ablation-tm} varies the Markovianity lookback horizon $T_M$ from 1 to 24~h. For DAS2 FS1, FS2, and FS4, the points remain close within each method variant. This indicates that changing $T_M$ has a smaller effect than enabling or disabling fallback. DAS2 FS3 shows a clearer difference between the fallback-enabled and forced-ICON variants, but the effect of $T_M$ within each variant is still limited.

Markovian 3000 is more sensitive to the lookback horizon. The complete SNF-ICON point moves from approximately 2.9~min and 5.76~MWh at $T_M=1$~h toward approximately 2.0~min and 6.04~MWh for longer windows. Thus, a longer window can reduce waiting time, but it may also increase energy use.

\Cref{fig:param-ablation-ab} varies the waiting-time weight $\alpha$ while keeping $\beta=1$. In general, increasing $\alpha$ gives more importance to waiting time, moving the results toward lower waiting and higher energy use. This trend is especially clear for Markovian 3000 and DAS2 FS4. Lower values of $\alpha$ favor energy savings, while higher values favor faster job starts. The best value therefore depends on the desired balance between waiting time and energy, rather than on one setting that works for every workload.

\begin{figure}[t]
    \centering
    \includegraphics[width=0.99\linewidth]{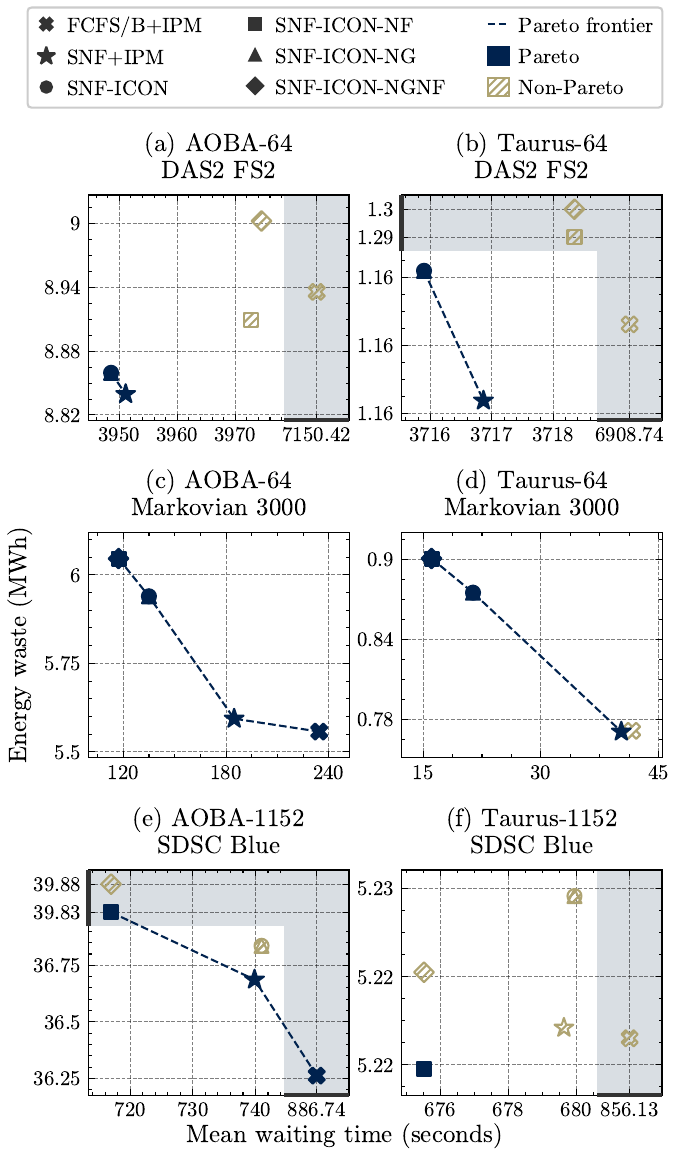}
    \caption{Cross-platform waiting-time--energy comparison on AOBA and Taurus configurations for three representative workloads. Dashed segments connect empirically nondominated points.}
    \label{fig:platform-new}
\end{figure}

\subsection{Cross-platform robustness}

\Cref{fig:platform-new} compares DAS2 FS2, Markovian 3000, and SDSC Blue on the AOBA and Taurus platform models. Changing the platform substantially changes the absolute waiting-time and energy scales. For example, SDSC Blue results use approximately 5.22~MWh on Taurus, compared with approximately 36.3--39.9~MWh on AOBA. Large differences also appear for the other workloads because the two platform models have different power use and node-transition behavior.

Even with these scale changes, the SNF-based methods remain in the lower-wait region across all six panels. The Markovian workload also shows a similar general trade-off on both platforms: configurations with lower waiting time tend to use more energy, while configurations with lower energy tend to have longer waits. This suggests that the main waiting-time benefit of SNF is not limited to one platform model.

However, the exact nondominated configuration changes across workloads and platforms, particularly for SDSC Blue. A gate or fallback setting that works well on AOBA may not give the same result on Taurus. These results show that SNF-ICON can operate across different platform models, but its parameters and operating variant should be selected using the power and transition characteristics of the target system.




\section{Conclusion}
\label{sec:conclusion}

This paper presented SNF-ICON, a two-mode method for energy-aware gang scheduling. At each scheduler invocation, the method checks recent job-arrival and completed-job data. If there are not enough data, or if the data do not fit the assumed model, the scheduler uses the SNF+IPM fallback. In this mode, queued jobs are still scheduled, nodes are still woken for planned jobs, and idle timeouts remain active. However, ICON-specific runtime predictions, next-event estimates, and warm-spare optimization are not used. The data are checked again at the next scheduler invocation, so the system can move between the two modes as the workload changes.

When the recent data pass the check, the scheduler uses ICON mode. Jobs that can run immediately are started first. The remaining queued jobs are then planned in SNF order using predicted node-release times, and sleeping nodes are scheduled to wake close to the predicted job start times. Warm-spare planning is considered only when no jobs remain in the queue and the scheduler is called by an initial, arrival, or completion event. The arrival-recency gate provides an additional check before warm-spare nodes are prepared. When this gate is closed, ICON remains active, but no spare nodes are prepared for future jobs. Completion-based shutdown and idle timeouts also prevent unused nodes from remaining active for too long.

The experiments show that SNF-ICON reduces average waiting time compared with FCFS/B+IPM in all six workload--platform cases. Its energy use remains close to the heuristic baselines in most cases, although the generated Markovian workload shows a noticeable energy increase. The workload screen selects ICON for much of the generated Markovian workload, while the DAS2 workloads operate mainly in fallback mode. The results also show that the best gate and fallback settings depend on the workload and platform. Therefore, checking recent workload behavior is useful for spare-node planning, but no single configuration gives the best waiting-time--energy trade-off in every case.

\section*{Acknowledgments}
This research is funded by the Indonesian Endowment Fund for Education (LPDP) on behalf of the Indonesian Ministry of Higher Education, Science and Technology and managed under the EQUITY Program (Contract Number: 4301/B3/DT.03.08/2025 and 10107/UN1.P/Dit-Keu/HK.08.00/2025).

\raggedbottom
\bibliographystyle{IEEEtran}
\balance
\bibliography{references}

\begin{IEEEbiographynophoto}{Reza Pulungan} received the bachelor's degree in computer science from the Universitas Gadjah Mada, Yogyakarta, Indonesia, in 1999, the master's degree in telematics from the Universiteit Twente, Enschede, The Netherlands, in 2002, and the Ph.D. degree in computer science from the Universit\"at des Saarlandes, Saarbr\"ucken, Germany, in 2009.

He is currently a Professor with the Department of Computer Science and Electronics, Universitas Gadjah Mada. His research background is in stochastic processes, especially Markov processes and phase-type distributions, and modeling and analyzing reactive systems. He is interested in and working on learning algorithms, including reinforcement learning, and their applications in diverse fields.
\end{IEEEbiographynophoto}

\begin{IEEEbiographynophoto}{Raka Satya Prasasta} received the bachelor's degree in informatics from the Universitas Ahmad Dahlan, Yogyakarta, Indonesia, in 2026. His research interests include HPC simulation frameworks, as well as optimization techniques in HPC systems, particularly energy-efficient job scheduling and dynamic power-state management. 
\end{IEEEbiographynophoto}

\begin{IEEEbiographynophoto}{Santana Yuda Pradata} received the bachelor's degree in computer science from Universitas Gadjah Mada, Yogyakarta, Indonesia, in 2026.

From August 2024 to January 2025, he was a Research Intern with the Research Center for Quantum Physics, National Research and Innovation Agency (BRIN), Indonesia. His research interests include quantum computing, quantum information, theoretical computer science, and optimization in computer science.
\end{IEEEbiographynophoto}

\begin{IEEEbiographynophoto}{Mursalim} received the master's degree in informatics from the Universitas Dian Nuswantoro, Semarang, Indonesia, in 2020. He is currently working toward the PhD degree in Computer Science with the Universitas Gadjah Mada. His research interests include machine learning and job scheduling on HPC systems.
\end{IEEEbiographynophoto}

\begin{IEEEbiographynophoto}{Hiroyuki Takizawa} is currently a professor and the deputy director of the Cyberscience Center, Tohoku University. His research interests include high-performance computing systems and their applications. His focus is particularly on the productivity in high-performance computing. He received the B.E. Degree in Mechanical Engineering, and the M.S. and Ph.D. Degrees in Information Sciences from Tohoku University in 1995, 1997 and 1999, respectively. He is a member of IEEE CS, ACM SIGHPC, IEICE and IPSJ.
\end{IEEEbiographynophoto}

\begin{IEEEbiographynophoto}{Muhammad Alfian Amrizal} is an Assistant Professor at the Department of Computer Science and Electronics, Universitas Gadjah Mada. He received his Ph.D. degree in Information Science from the Graduate School of Information Sciences, Tohoku University, in 2017. His main research interests are in the area of distributed systems, such as high-performance computing (HPC) systems and wireless sensor networks (WSN), including dependability and energy efficiency, novel fault tolerance techniques, performance modeling, and optimization of such systems. He is also interested in broad topics of optimization problems and AI.
\end{IEEEbiographynophoto}

\clearpage

\appendices

\section{Job-Specific Runtime and Release Prediction}
\label[appendix]{app:runtime-prediction}

\begin{assumption}[Shared multiplicative runtime error]
\label{ass:runtime-ratio}
For job $j$, let $\tau_j$ be its requested runtime and let $D_j>0$ be its actual total runtime. The predictor models the ratio $D_j/\tau_j$ using a distribution shared across the workload. It assumes that this distribution changes slowly enough for its log-scale mean and variance to be tracked using exponential moving averages.
\end{assumption}

For a completed job $j$, let $D_j^{\mathrm{obs}}$ be its observed runtime (execution time). This value is also used as a sample in the service screen. The predictor uses the log ratio
\begin{equation}
Y_j=\log\left(\frac{D_j^{\mathrm{obs}}}{\tau_j}\right)
\label{eq:log-ratio}
\end{equation}
to measure the difference between the actual and requested runtimes. Requested runtimes are standard scheduler inputs, but history-based corrections can improve their accuracy \cite{utilpred,tsafrir2007}.

The estimates are initialized from the first valid completed-job observation by setting $m_Y=Y_1$ and $q_Y=Y_1^2$. For each subsequent valid completed-job observation, the first and second moments of the log ratio are updated using the service smoothing factor $\rho_S$:
\begin{align}
m_Y^{+}
&=
\rho_S Y_j+(1-\rho_S)m_Y,
\label{eq:ratio-first-moment}\\
q_Y^{+}
&=
\rho_S Y_j^2+(1-\rho_S)q_Y.
\label{eq:ratio-second-moment}
\end{align}
After each update, $m_Y$ and $q_Y$ are replaced by $m_Y^{+}$ and $q_Y^{+}$.

The estimated median actual-to-requested runtime ratio is
\begin{equation}
\widehat{\eta}_{\mathrm{med}}=\exp(m_Y),
\label{eq:ratio-median}
\end{equation}
and the estimated standard deviation of the log ratio is
\begin{equation}
\widehat{\sigma}_Y
=
\max\left\{
\sigma_{\min},
\sqrt{\max\{0,q_Y-m_Y^2\}}
\right\},
\label{eq:ratio-sigma}
\end{equation}
where $\sigma_{\min}>0$ prevents the estimated variance from becoming zero.

For a queued job, the predicted total runtime is
\begin{equation}
\widehat{D}_j
=
\tau_j\widehat{\eta}_{\mathrm{med}}.
\label{eq:queued-duration}
\end{equation}

For an active job, let $e_j(t)\ge0$ be its elapsed runtime at time $t$. All durations inside logarithms are expressed as positive values in seconds. The total runtime is modeled as log-normal \cite{corlett1957lognormal}:
\begin{equation}
\log D_j
\sim
\mathcal{N}\!\left(
\theta_j,\widehat{\sigma}_Y^{\,2}
\right),
\qquad
\theta_j=\log\tau_j+m_Y.
\label{eq:active-lognormal}
\end{equation}

For $e_j(t)>0$, define
\begin{equation}
\zeta_j(t)
=
\frac{\log e_j(t)-\theta_j}{\widehat{\sigma}_Y},
\end{equation}
and let
\begin{equation}
\overline{\Phi}(u)=1-\Phi(u)
\end{equation}
be the standard-normal survival function.

Because an active job has already run for more than $e_j(t)$, its predicted remaining runtime is
\begin{align}
\widehat{R}_j(t)
&=
\E[D_j-e_j(t)\mid D_j>e_j(t)],
\nonumber\\
&=
\exp\!\left(
\theta_j+\frac{\widehat{\sigma}_Y^{\,2}}{2}
\right)
\frac{
\overline{\Phi}\!\left(
\zeta_j(t)-\widehat{\sigma}_Y
\right)
}{
\overline{\Phi}\!\left(
\zeta_j(t)
\right)
}
-e_j(t).
\label{eq:conditional-remaining}
\end{align}

When $e_j(t)=0$, the predictor uses the unconditional mean
\begin{equation}
\exp\!\left(
\theta_j+\frac{\widehat{\sigma}_Y^{\,2}}{2}
\right).
\end{equation}

The queued-job estimate in~\cref{eq:queued-duration} uses the median of the predicted runtime distribution. The active-job estimate uses the conditional mean after accounting for the time already elapsed. These are different estimates and are used for different cases.

To derive~\cref{eq:conditional-remaining}, first write
\begin{equation}
\E[D_j\mid D_j>e_j(t)]
=
\frac{
\E[D_j\mathbf{1}\{D_j>e_j(t)\}]
}{
\Prob(D_j>e_j(t))
}.
\label{eq:conditional-lognormal-start}
\end{equation}

Let
\begin{equation}
V
=
\frac{\log D_j-\theta_j}{\widehat{\sigma}_Y}
\sim\mathcal{N}(0,1).
\end{equation}
The denominator in~\cref{eq:conditional-lognormal-start} is
\begin{equation}
\Prob(D_j>e_j(t))
=
\overline{\Phi}(\zeta_j(t)).
\end{equation}

The numerator is
\begin{align}
\E[D_j\mathbf{1}\{D_j>e_j(t)\}]
&=
\int_{\zeta_j(t)}^{\infty}
\exp\!\left(
\theta_j+\widehat{\sigma}_Y v
\right)
\frac{e^{-v^2/2}}{\sqrt{2\pi}}
\,\mathrm{d}v,
\nonumber\\
&=
\exp\!\left(
\theta_j+\frac{\widehat{\sigma}_Y^{\,2}}{2}
\right)
\nonumber\\
&\quad\times
\int_{\zeta_j(t)}^{\infty}
\frac{
e^{-(v-\widehat{\sigma}_Y)^2/2}
}{
\sqrt{2\pi}
}
\,\mathrm{d}v,
\nonumber\\
&=
\exp\!\left(
\theta_j+\frac{\widehat{\sigma}_Y^{\,2}}{2}
\right)
\overline{\Phi}\!\left(
\zeta_j(t)-\widehat{\sigma}_Y
\right).
\label{eq:conditional-lognormal-numerator}
\end{align}

Substituting the numerator and denominator into~\cref{eq:conditional-lognormal-start}, and then subtracting the elapsed runtime $e_j(t)$, gives~\cref{eq:conditional-remaining}. The implementation computes the survival-function ratio in log space to avoid numerical problems.

The predicted finish time of an active job is
\begin{equation}
t+\widehat{R}_j(t),
\end{equation}
while the predicted finish time of a queued job planned to start at $s_j$ is
\begin{equation}
s_j+\widehat{D}_j.
\end{equation}
These values are used in the release map.

\section{Expected Unfinished Transition Time}
\label[appendix]{app:residual-transition}

\begin{lemma}[Expected unfinished transition time]
\label{lem:residual-transition}
Suppose a node transition has $d\ge0$ units of time remaining. Let the modeled time to the next event be
\begin{equation}
T\sim\mathrm{Exp}(\widehat{\lambda}_E(t)),
\qquad
\widehat{\lambda}_E(t)>0.
\end{equation}
The expected transition time still remaining when the event occurs is
\begin{equation}
\phi(d;\widehat{\lambda}_E(t))
=
\E\!\left[\pos{d-T}\right]
=
d-
\frac{
1-e^{-\widehat{\lambda}_E(t)d}
}{
\widehat{\lambda}_E(t)
}.
\label{eq:residual}
\end{equation}
\end{lemma}

\begin{proof}
The unfinished transition time is positive only when the next event occurs before the transition finishes, that is, when $T<d$. Using the tail-integral identity,
\begin{align}
\E[\pos{d-T}]
&=
\int_0^d
\Prob(d-T>u)
\,\mathrm{d}u,
\nonumber\\
&=
\int_0^d
\Prob(T<d-u)
\,\mathrm{d}u,
\nonumber\\
&=
\int_0^d
\left(
1-e^{-\widehat{\lambda}_E(t)(d-u)}
\right)
\,\mathrm{d}u.
\end{align}

Using the change of variable $v=d-u$ gives
\begin{align}
\E[\pos{d-T}]
&=
\int_0^d
\left(
1-e^{-\widehat{\lambda}_E(t)v}
\right)
\,\mathrm{d}v,
\nonumber\\
&=
d-
\frac{
1-e^{-\widehat{\lambda}_E(t)d}
}{
\widehat{\lambda}_E(t)
}.
\end{align}
This proves~\cref{eq:residual}.
\end{proof}

\end{document}